\documentclass[a4paper,USenglish]{article} 

\usepackage{amsmath}
\usepackage{amsthm}
\usepackage{amssymb}
\usepackage{microtype}
\usepackage{complexity}
\usepackage[ruled,vlined]{algorithm2e}			
\usepackage
	[noend]		
	{algorithmic}

\usepackage{doi}
\usepackage{fullpage}
\usepackage{xspace}
\usepackage[colorinlistoftodos]{todonotes}
\usepackage{stmaryrd} 
\usepackage{enumerate}
\usepackage{authblk}
\usepackage{xfrac}
\usepackage[format=hang,justification=justified]{subcaption}
\usepackage{booktabs}
\usepackage{mathtools}

\usepackage{tikz} 
\usepackage{thmtools,thm-restate}

\newtheorem{theorem}{Theorem}
\newtheorem{lemma}[theorem]{Lemma}
\newtheorem{definition}[theorem]{Definition}
\newtheorem{proposition}[theorem]{Proposition}

\newtheorem{corollary}[theorem]{Corollary}
\theoremstyle{definition}

\newtheorem{myclaim}{Claim}

\newcommand{\Z}{\mathbf{Z}}

\newcommand{\nat}{\mathbb{N}}

\newcommand*{\krst}{\textsc{$k$-FGC}\xspace}
\newcommand*{\onerst}{\textsc{FGC}\xspace}

\newcommand*{\I}{\ensuremath{I}\xspace}

\newcommand*{\formatmathnames}[1]{\textnormal{\small #1}}
\newcommand*{\OPT}{\formatmathnames{OPT}}

\newcommand*{\algo}{\ensuremath{\mathcal{A}}\xspace}
\newcommand*{\algoA}{A\xspace}
\newcommand*{\algoB}{B\xspace}
\newcommand*{\algoC}{C\xspace}
\newcommand*{\groundset}{\ensuremath{X}\xspace}
\newcommand*{\Alphas}{W\xspace}

\newcommand{\unsafe}{\text{unsafe}\xspace}
\newcommand{\safe}[1]{\ensuremath{\overline{#1}}}
\newcommand{\sce}{\ensuremath{E'}}
\newcommand{\WTAP}{\textsc{WTAP}\xspace}
\newcommand{\utap}{\textsc{TAP}\xspace}
\newcommand{\ECSS}{2-\textsc{ECSS}\xspace}

\newcommand{\flex}{\textsc{FGC}\xspace}
\newcommand{\tecs}{\textsc{$2$-ECSS}\xspace}
\newcommand{\tap}{\textsc{WTAP}\xspace}
\newcommand{\boundtwo}{\ensuremath{2.523}\xspace}

\newcommand{\boundthreehalf}{\ensuremath{2.4036}\xspace}
\newcommand{\boundthreehalflow}{\ensuremath{2.4035}\xspace}

\author[1]{David Adjiashvili}
\author[2]{Felix Hommelsheim\thanks{Research partially supported by the German Research Foundation (DFG), RTG 1855}}
\author[3]{Moritz M\"uhlenthaler}
\affil[1]{Department of Mathematics,  ETH Z\"urich, Switzerland }
\affil[2]{Fakult\"at f\"ur Mathematik, TU Dortmund University, Germany}
\affil[3]{Laboratoire G-SCOP, Grenoble INP, Univ.~Grenoble-Alpes, France}

\date{}
\title{Flexible Graph Connectivity: Approximating Network Design Problems Between 1- and 2-connectivity}

\usetikzlibrary{decorations.pathmorphing}
\usetikzlibrary{decorations.pathreplacing} 

\tikzset{
	edge/.style={thick, gray},
	medge/.style={decorate,very thick,decoration={snake}},
	aedge/.style={very thick,dashed,black},
	dedge/.style={thick,->},
	availedge/.style={thick,blue},
	vertex/.style={shape=circle,thick,draw,node distance=3em}
}

\begin{document} 
   
\maketitle

\begin{abstract}
Graph connectivity and network design problems are among the most fundamental problems in combinatorial optimization. The minimum spanning tree problem, the two edge-connected spanning subgraph problem (\tecs) and the tree augmentation problem (\tap) are all examples of fundamental well-studied network design tasks that postulate different initial states of the network and different assumptions on the reliability of network components.
In this paper we motivate and study \emph{Flexible Graph Connectivity} (\flex), a problem that mixes together both the modeling power and the complexities of all aforementioned problems and more. 
In a nutshell, \flex asks to design a connected network, while allowing to
specify different reliability levels for individual edges. 
While this non-uniform nature of the problem makes it appealing from  the modeling perspective, it also renders most existing algorithmic tools for dealing with network design problems unfit for approximating \flex.

In this paper we develop a general algorithmic approach for approximating \flex that yields approximation algorithms with ratios that are close to the known best bounds for many special cases, such as \tecs and \tap. Our algorithm and analysis combine various techniques including a weight-scaling algorithm, a charging argument that uses a variant of exchange bijections between spanning trees and a factor revealing min-max-min optimization problem.
\end{abstract}

\newpage
\tableofcontents
\newpage

\renewcommand{\arraystretch}{1.25}

\section{Introduction}
\label{sec:introduction}

Many real-world design and engineering problems can be modeled as either graph connectivity or network design problems. Routing, city planning, communication infrastructure are just a few examples where network design problems are omnipresent. To realistically model a real-life problem as a network design problem, it is imperative to consider issues of \emph{reliability}, namely the capacity of the systems' resources to withstand disturbances, failures, or even adversarial attacks. This aspect motivates the study of many classical network design problem, as well as robust counterparts of many connectivity problem. 

The problem studied in this paper, which we call \emph{Flexible Graph
Connectivity} (\flex), lies in the intersection of classical network design and
robust optimization. It encapsulates several well studied problems that have
received significant attention from the research community, such as the minimum
spanning tree problem, the 2-edge-connected spanning subgraph problem
(\tecs)~\cite{frederickson_jaja_81,Jain2001,Goemans:94,SV:14,ccivril20195}, the weighted
tree augmentation problem
(\tap)~\cite{adjiashvili2018beating,fiorini2018approximating,frederickson_jaja_81,grandoni2018improved,kortsarz2018lp,kortsarz2016simplified,nutov2017tree}, and the matching augmentation problem~\cite{MAP:18}.
As such, \flex is \APX-hard and it encompasses all of the technical challenges
associated with approximating these problems simultaneously. 
In a sense, minimum spanning tree and \tecs represent two far ends of a 
spectrum of possible network design tasks that can be modeled with \flex and
the other mentioned problems lie in between. We argue that by translating
attributes of a real-life network design problem, one is much more likely to
encounter a problem from the aforementioned spectrum, as opposed to one of its
more famous extreme cases.

The problem \flex is formally defined as follows. The input is given by an undirected connected graph $G = (V,E)$, non-negative edge weights $w \in \mathbb{Q}_{\geq 0}^E$ and a set of edges $\overline F \subseteq E$ called \emph{safe} edges. Let $F := E\setminus \overline F$ be the \emph{\unsafe} edges. The task is to compute a minimum-weight edge set $S\subseteq E$ with the property that $(V,S-f)$ is a connected graph for every $f\in F$. 

We briefly illustrate why minimum spanning tree, \tecs, and \tap are all special cases of \flex.
Clearly, if all edges of the input graph are safe, then an optimal solution is a minimum-weight spanning tree. 
If, on the other hand, all edges of the input graph are unsafe, then an optimal solution is a minimum-weight 2-edge-connected spanning subgraph.
Finally, if the \unsafe edges form a weight-zero spanning tree $T$ of the input graph, then an optimal solution is a minimum-cost tree augmentation of $T$. 
The goal of this paper is to provide approximation algorithms for \flex and our main result is the following theorem.

\begin{theorem}
    \label{thm:flex}
    \flex admits a polynomial-time  $\boundtwo$-approximation algorithm.
\end{theorem}

In the spirit of recent advancements for \tap our results also extend to bounded-weight versions of the problem. A bounded-weight \flex instance is one whose weights all range between $1$ and some fixed constant $M \in \nat$. For bounded-weight \flex we obtain the following result. 

\begin{theorem}
    \label{thm:flexbounded}
    Bounded-weight
    \flex admits a polynomial-time $\boundthreehalf$-approximation algorithm.
  \end{theorem}
 
We elaborate on the techniques used to prove these theorems as well as their connection to results on \tecs and \tap later on. We start by giving our motivation for studying this problem.

\subsection{Importance of non-uniform models for network reliability}

The vast majority of network design problems are motivated by reliability
requirements imposed on real-life networks. For example, the $k$-edge connected
spanning subgraph problem~\cite{gabow_et_al_kECSS_09,Gabow:12} asks to
construct a connected network that can withstand a failure of at most $k-1$
edges. In the \emph{$k$-edge connectivity augmentation problem} we
are given a $k$-edge connected graph and the goal is to add a minimum
weight-set of edges such that the resulting graph is ($k+1$)-edge connected. It
is shown in \cite{dinits1976structure} that $k$-edge connectivity augmentation can be reduced to \tap for odd $k$ and to cactus
augmentation for even $k$~\cite{byrka2019breaching,galvez2019cycle}.
Thus, the recent result for cactus augmentation~\cite{byrka2019breaching} implies a $1.91$-approximation for unweighted $k$-edge connectivity augmentation (together with the $1.5$-approximation~\cite{grandoni2018improved} for \utap).
The more general \emph{survivable network design problem}
asks to construct a network that admits a prescribed number of edge-disjoint
paths between every pair of nodes~\cite{Goemans:94,Jain2001}.

While all of the above problems are important for modeling reliability in the design of real-world networks, they also neglect the inherent inhomogeneity of resources (such as nodes, links etc.) in such environments. This inhomogeneity stems from various factors, such as geographical region, available construction material, proximity to hazards, the ability to defend the asset and more. To incorporate this aspect of real-world problems it is imperative to differentiate between different resources, not only in terms of their cost, but also in terms of their reliability. In the simplest setup, we would like to distinguish between \emph{safe} and \emph{\unsafe} resources (edges, nodes etc.) and postulate that failures can only occur among the \unsafe resources. With \flex we adopt this model and study the basic graph connectivity problem.
 
In the robust optimization literature, several non-uniform failure models have
been proposed. Adjiashvili, Stiller and Zenklusen~\cite{bulk} proposed the
bulk-robust model, in which a solution has to be chosen that withstands the
failure of any input prescribed set of scenarios, each comprising a subset of
the resources. Since subsets can be specified arbitrarily, bulk-robustness can
be used to model high correlations between failures of individual elements, as
well as highly non-uniform scenarios. \flex falls into the category of bulk-robust network
design problems, since we can model the reliability criterion by creating
failure scenarios, one per unsafe edge in the bulk-robust setup and require the
graph to be connected. In fact, the existing results~\cite{bulk} on bulk-robust
optimization imply the existence of $\log n$ ratio for the problem. In this
paper we improve this significantly.

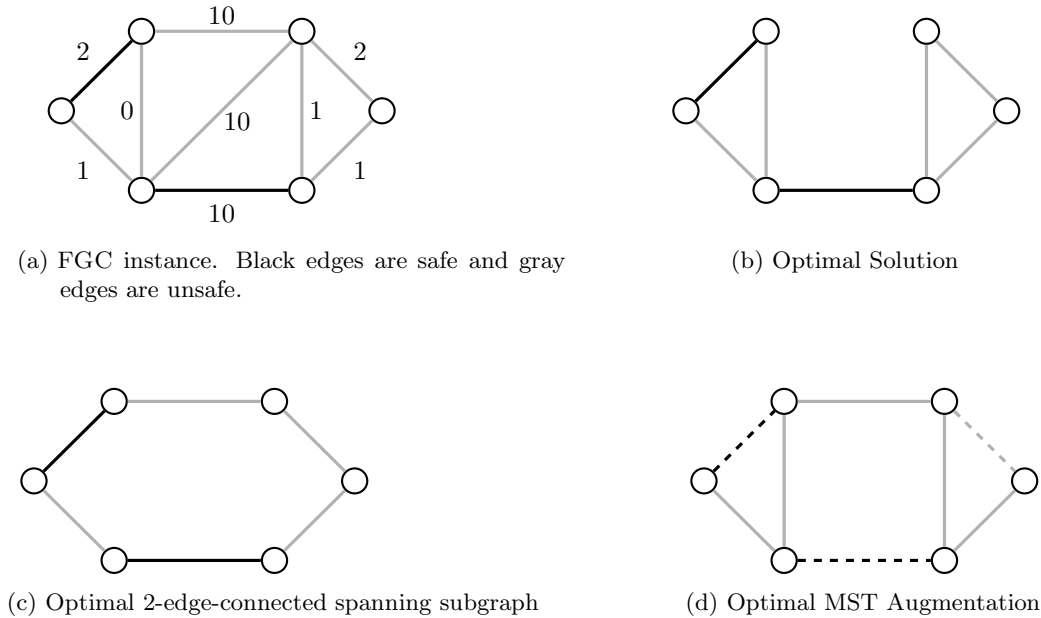
\begin{figure}[t]
 \begin{center}
    \begin{subfigure}[t]{.45\linewidth}
      \centering
      \begin{tikzpicture}[vertex/.style={shape=circle,thick,draw,node distance=3em,fill=white}]
        \node[vertex,label=left:] (v1) {};
        \node[vertex,right of =v1, draw=white] (a) {};
        \node[vertex,above of =a] (v2) {};
	\node[vertex,right of =a, draw=white] (b) {};
        \node[vertex,below of =a] (v3) {};
	\node[vertex,right of =b, draw=white] (c) {};
        \node[vertex,above of =c] (v4) {};
        \node[vertex,right of =c] (v5) {};
        \node[vertex,below of =c] (v6) {};

	\node[vertex,right of =v5, draw=white] (d10) {};
	\node[vertex,right of =d10, draw=white] (d) {};
        \draw[-,very thick, color=black] (v1) -- (v2) node [midway,
	  color=black, label ={[label distance=-0.5em]above left: \textcolor{black}{$2$}}] {};
        \draw[-,very thick, color=black!30!white] (v1) -- (v3) node [midway,
	  color=black, label={[label distance=-0.5em]below left: \textcolor{black}{$1$}}] {};
        \draw[-,very thick, color=black!30!white] (v2) -- (v3) node [midway,
color=black, label ={[label distance=-0.5em]left:\textcolor{black}{$0$}}] {};
	\draw[-,very thick, color=black] (v3) to  node [midway,label={[label distance=-2em]:\textcolor{black}{$10$}}] {} (v6);
        \draw[-,very thick, color=black!30!white] (v3) -- (v4) node [midway,
	  label ={[label distance=-1em]below right: \textcolor{black}{$10$}}] {};
        \draw[-,very thick, color=black!30!white] (v2) -- (v4) node [midway,
	  label ={[label distance=-0.5em]above: \textcolor{black}{$10$}}] {};
        \draw[-,very thick, color=black!30!white] (v6) -- (v5) node [midway,
	  label = {[label distance=-0.5em]below right: \textcolor{black}{$1$}}] {};
        \draw[-,very thick, color=black!30!white] (v4) -- (v5) node [midway,
	  label = {[label distance=-0.5em]above right: \textcolor{black}{$2$}}] {};
        \draw[-,very thick, color=black!30!white] (v6) -- (v4) node [midway,
	  label = {[label distance=-0.5em]right: \textcolor{black}{$1$}}] {};
	\end{tikzpicture}
	\caption{\flex instance. Black edges are safe and gray edges are \unsafe.\label{fig:example:instance}}
    \end{subfigure}
    \begin{subfigure}[t]{.45\linewidth}
      \centering
      \begin{tikzpicture}[vertex/.style={shape=circle,thick,draw,node distance=3em,fill=white}]
        \node[vertex,right of =d] (u1) {};
        \node[vertex,right of =u1, draw=white] (e) {};
        \node[vertex,above of =e] (u2) {};
	\node[vertex,right of =e, draw=white] (f) {};
        \node[vertex,below of =e] (u3) {};
	\node[vertex,right of =f, draw=white] (g) {};
        \node[vertex,above of =g] (u4) {};
        \node[vertex,right of =g] (u5) {};
        \node[vertex,below of =g] (u6) {};

        \draw[-,very thick, color=black] (u1) -- (u2);
        \draw[-,very thick, color=black!30!white] (u1) -- (u3);
        \draw[-,very thick, color=black!30!white] (u2) -- (u3);
        \draw[-,very thick, color=black] (u3) -- (u6);

        \draw[-,very thick, color=black!30!white] (u6) -- (u5);
        \draw[-,very thick, color=black!30!white] (u4) -- (u5);
        \draw[-,very thick, color=black!30!white] (u6) -- (u4);

	\phantom{\draw[-,very thick, color=black] (u3) to  node [midway,label={[label distance=-2em]:\textcolor{black}{$10$}}] {} (u6);}
      \end{tikzpicture}
      \caption{Optimal Solution\label{fig:example:opt}}
    \end{subfigure}
    \\
    \begin{subfigure}[t]{.48\linewidth}
      \centering
      \begin{tikzpicture}[vertex/.style={shape=circle,thick,draw,node distance=3em}]
	\node[vertex,below of =v3, fill=white, draw=white] (z) {};

        \node[vertex,below of =z, fill=white] (w2) {};
        \node[vertex, below of =w2, fill=white, draw=white] (a1) {};
        \node[vertex,left of =a1, fill=white] (w1) {};
	\node[vertex,right of =a1, fill=white, draw=white] (b1) {};
        \node[vertex,below of =a1, fill=white] (w3) {};
	\node[vertex,right of =b1, fill=white, draw=white] (c1) {};
        \node[vertex,above of =c1, fill=white] (w4) {};
        \node[vertex,right of =c1, fill=white] (w5) {};
        \node[vertex,below of =c1, fill=white] (w6) {};

	\node[vertex,right of =w5, fill=white, draw=white] (d11) {};
	\node[vertex,right of =d11, fill=white, draw=white] (d1) {};
        \draw[-,very thick, color=black] (w1) -- (w2);
        \draw[-,very thick, color=black!30!white] (w1) -- (w3);
        \draw[-,very thick, color=black] (w3) -- (w6);
        \draw[-,very thick, color=black!30!white] (w2) -- (w4);
        \draw[-,very thick, color=black!30!white] (w6) -- (w5);
        \draw[-,very thick, color=black!30!white] (w4) -- (w5);
      \end{tikzpicture}
      \caption{Optimal 2-edge-connected spanning subgraph\label{fig:example:2ecss}}
    \end{subfigure}
    \begin{subfigure}[t]{.48\linewidth}
      \centering
      \begin{tikzpicture}[vertex/.style={shape=circle,thick,draw,node distance=3em}]
        \node[vertex,right of =d1, fill=white] (x1) {};
        \node[vertex,right of =x1, fill=white, draw=white] (e1) {};
        \node[vertex,above of =e1, fill=white] (x2) {};
	\node[vertex,right of =e1, fill=white, draw=white] (f1) {};
        \node[vertex,below of =e1, fill=white] (x3) {};
	\node[vertex,right of =f1, fill=white, draw=white] (g1) {};
        \node[vertex,above of =g1, fill=white] (x4) {};
        \node[vertex,right of =g1, fill=white] (x5) {};
        \node[vertex,below of =g1, fill=white] (x6) {};

        \draw[-,very thick, color=black, dashed] (x1) -- (x2);
        \draw[-,very thick, color=black!30!white] (x1) -- (x3);
        \draw[-,very thick, color=black!30!white] (x2) -- (x3);
        \draw[-,very thick, color=black, dashed] (x3) -- (x6);
        \draw[-,very thick, color=black!30!white] (x2) -- (x4);
        \draw[-,very thick, color=black!30!white] (x6) -- (x5);
        \draw[-,very thick, color=black!30!white,dashed] (x4) -- (x5);
        \draw[-,very thick, color=black!30!white] (x6) -- (x4);
      \end{tikzpicture}
      \caption{Optimal MST Augmentation}
    \end{subfigure}
  \end{center}
 \caption{The versatility of \flex. In (a) an instance of \flex, with \unsafe edges colored gray. In (b) the optimal solution. In (c) the optimal solution to the same instance, except where all edges are \unsafe. The problem coincides with \tecs on the same graph. In (d) the dashed edges are an optimal \flex, solution for the instance where the gray edges have cost zero and are the \unsafe edges. It is equivalent to the tree augmentation problem for the gray minimum spanning tree.}\label{fig:example}
\end{figure}

\subsection{On the complexity of \flex and its relationship to \tecs and \tap}

As was pointed out before, some classical and well studied network design problems are special cases of \flex, including \tecs, and \tap. Thus, approximating \flex, is at least as challenging as approximating the latter two problems. In this section we present some evidence that \flex might actually be significantly harder. For the general case of both \tecs and \tap the iterative rounding algorithm of Jain~\cite{Jain2001} provides an approximation factor of $2$, which is best known. It is important to note that \tecs subsumes \tap in the general (weighted) case. Nevertheless, it is both instructive and useful to relate \flex to both problems.
In particular, this enables us to improve the approximation ratio for the bounded-weight version of \flex.

In contrast, for the unweighted versions of both problems (and in the case of
\tap more generally for bounded weights), a long line of results has generated
numerous improvements beyond ratio two, leading to the currently best known
bounds of $5/4$ for unweighted \tecs~\cite{ccivril20195}, $1.46$ for unweighted tree
augmenation~\cite{grandoni2018improved} and $1.5$ for \tap with
bounded weights~\cite{fiorini2018approximating,grandoni2018improved}. The case
of unit (or bounded) weights is where techniques for approximating \tecs
and \tap start to differ significantly. In both cases, the known best bounds
are achieved by combining LP-based techniques with clever combinatorial tools.
Nevertheless, there seems to be very little intersection in both the nature of
used LPs and the overall approaches, as techniques suitable for one problem do
not seem to provide competitive ratios for the other.

Consequently, there are several implications for approximating \flex. Firstly,
achieving an approximation factor better than two is an ambitious task, as it
would simultaneously improve the long-standing best known bounds for both \tecs
and \tap. At the same time, for achieving a factor two, it may be possible to
use classical tools for survivable network design~\cite{Jain2001}. We show
that, at least with the natural LPs, this is impossible, as the integrality gap
of such LPs can be significantly larger than $2$. Consider the
following natural generalization of the cut-based formulation for survivable
network design to \flex. 

\begin{equation}
  \begin{aligned}
    \text{minimize }\; & w^T x \\
    \text{subject to }\; & \sum_{f \in \delta(S) \cap F} x_f + \sum_{e \in \delta(S) \cap \overline {F}} 2 x_e \geq 2\; \text{ for all } \emptyset \subsetneq S \subsetneq V\\
    & x_e \in \{0, 1\}\; \text{ for all } e \in E.
  \end{aligned}
  \label{eq:ilp}
\end{equation}

In essence, the IP formulation~\eqref{eq:ilp} states that each cut in the
graph needs to contain at least one safe edge or at least two edges, which
indeed is the feasibility condition for \flex.  One can show that many
important properties that are central in Jain's~\cite{Jain2001} analysis still
hold, e.g., the possibility to perform uncrossing for tight constraints at a
vertex LP solution. These properties might become useful for devising a pure
LP-based algorithm for \flex. However, for the instance shown in
Figure~\ref{fig:integrality-gap}, the integrality gap of formulation~\eqref{eq:ilp} is at
least $8/3$.  

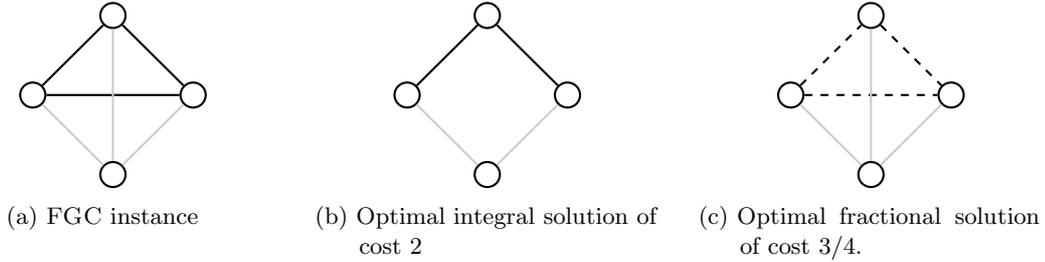
\begin{figure}
  \begin{center}
    \begin{subfigure}[t]{.28\linewidth}
      \centering
      \begin{tikzpicture}[vertex/.style={shape=circle,thick,draw,node distance=3em,fill=white}]
	\node[vertex,label=left:] (v1) {};
	\node[vertex, right of =v1, draw=white] (a) {};
	\node[vertex,label=left:, above of =a] (v2) {};
	\node[vertex,label=left:, right of =a] (v3) {};
	\node[vertex,label=left:, below of =a] (v4) {};

	\draw[-,thick, color=black] (v1) -- (v2);
	\draw[-,thick, color=black] (v1) -- (v3);
	\draw[-,thick, color=black] (v2) -- (v3);
	\draw[-,thick, color=black!20!white] (v1) -- (v4);
	\draw[-,thick, color=black!20!white] (v2) -- (v4);
	\draw[-,thick, color=black!20!white] (v3) -- (v4);
      \end{tikzpicture}
      \caption{\flex instance}
    \end{subfigure}
    \hspace{1em}
    \begin{subfigure}[t]{.28\linewidth}
      \centering
      \begin{tikzpicture}[vertex/.style={shape=circle,thick,draw,node distance=3em,fill=white}]
	\node[vertex,fill=white] (u1) {};
	\node[vertex, right of =u1, draw=white] (b) {};
	\node[vertex,label=left:, above of =b] (u2) {};
	\node[vertex,label=left:, right of =b] (u3) {};
	\node[vertex,label=left:, below of =b] (u4) {};	

	\draw[-,thick, color=black] (u1) -- (u2);
	\draw[-,thick, color=black] (u2) -- (u3);
	\draw[-,thick, color=black!20!white] (u1) -- (u4);
	\draw[-,thick, color=black!20!white] (u3) -- (u4);
      \end{tikzpicture}
      \caption{Optimal integral solution of cost 2}
    \end{subfigure}
    \hspace{1em}
    \begin{subfigure}[t]{.28\linewidth}
      \centering
      \begin{tikzpicture}[vertex/.style={shape=circle,thick,draw,node distance=3em,fill=white}]
	\node[vertex] (w1) {};
	\node[vertex, right of =w1, draw=white] (c) {};
	\node[vertex,label=left:, above of =c] (w2) {};
	\node[vertex,label=left:, right of =c] (w3) {};
	\node[vertex,label=left:, below of =c] (w4) {};	        

	\draw[-,thick, color=black, dashed] (w1) -- (w2);
	\draw[-,thick, color=black, dashed] (w1) -- (w3);
	\draw[-,thick, color=black, dashed] (w2) -- (w3);
	\draw[-,thick, color=black!20!white] (w1) -- (w4);
	\draw[-,thick, color=black!20!white] (w2) -- (w4);
	\draw[-,thick, color=black!20!white] (w3) -- (w4);
      \end{tikzpicture}
      \caption{Optimal fractional solution of cost $3/4$.}
    \end{subfigure}
    \caption{\flex instance for which~\eqref{eq:ilp} has integrality gap $8/3$. Gray edges are unsafe and
      have cost 0. Black edges are safe and have cost 1. The dashed (safe)
      edges are fractional with value $1/4$.\label{fig:integrality-gap}}
    \end{center}
\end{figure}

Finally, the recent advances achieved for unweighted and bounded-weight versions of \tecs and \tap seem to be unsuitable to directly tackle \flex, as they were not found to provide good ratios for both \tecs and \tap simultaneously. It is natural to conclude that a good ratio for \flex can only be achieved by a combination of techniques suitable for both \tecs and \tap. 

To summarize, it seems that while the only known techniques for simultaneously approximating both \tecs and \tap within a factor two rely on rounding natural linear programs relaxations of a more general network design problem (such as the survivable network design problem~\cite{Goemans:94,Jain2001}), the integrality gaps of such natural LPs for \flex are significantly larger than two. At the same time, due to its strong motivation, it is desirable to achieve a factor close to two, which is the state of the art for both \tecs and \tap. 

In this paper we show that this goal can be achieved by properly combining algorithms for \tecs, and \tap. Our algorithms are simple and black-box to such an extent that results for restricted versions of \tap (e.g., bounded cost) can be directly applied to \flex with the same restrictions, thus leading to improved bounds for these restricted versions of \flex as well. 
At the same time, the analysis is complex and requires careful charging arguments, generalizations of the notion of exchange bijections of spanning trees and factor revealing optimization problems, as we elaborate next.

\subsection{Main techniques and an overview of the algorithm}

\begin{table}[t]
    \centering
    \caption{Approximation ratios for \flex and some special cases.\label{tab:apx:overview}}
    \begin{tabular}{r|cccl}
        \toprule
    problem & general weights & bounded weights & unweighted \tabularnewline[0.25em]
        \midrule
        \flex & \boundtwo (Thm.~\ref{thm:flex}) & \boundthreehalf (Thm.~\ref{thm:flexbounded}) & 23/16 (Thm.~\ref{thm:unweightedkflex}) & \tabularnewline
        \tecs & 2~\cite{Jain2001}   & 2~\cite{Jain2001} & 5/4~\cite{ccivril20195} &\tabularnewline
        \tap  & 2~\cite{frederickson_jaja_81}   & $23/16$~\cite{grandoni2018improved} &  1.46~\cite{grandoni2018improved} &\tabularnewline
        \bottomrule
    \end{tabular}
\end{table}

We present here a high-level overview of some of the technical ingredients that  go into our algorithm and analysis used to prove Theorem~\ref{thm:flex}. The algorithm carefully combines the following three rather simple algorithms for \flex, each having an approximation ratio significantly worse than the one exhibited in Theorem~\ref{thm:flex}. 
A more detailed description of the algorithms will be given later on.
The main idea is to combine in a clever way the following three simple algorithms, which we formally state later on.
Each of the algorithms establishes 2-edge-connectivity in a modified graph, where
a subset of safe edges is contracted.
In order to be able to establish 2-edge connectivity, we need to add a parallel
\unsafe edge $e'$ for each safe edge $e$ of the same cost.  
It can easily be observed that optimal solutions
for the new instance are also optimal solution for the old instance
and vice versa. 
The three algorithms are stated below.
\begin{center}
    \begin{tabular}[h]{rp{0.75\linewidth}}
        \textbf{Algorithm A} & Compute a 2-edge connected spanning subgraph.\tabularnewline
        \textbf{Algorithm B} & Compute a minimum spanning tree, contract its safe edges and make it 2-connected by solving the corresponding \tap instance. \tabularnewline
        \textbf{Algorithm C} & Compute a minimum spanning tree, contract its safe edges, then compute a 2-edge connected spanning subgraph. Return the union of this solution and the safe edges of the spanning tree.\tabularnewline
    \end{tabular}
\end{center}

It is not hard to show that Algorithms A, B and C are polynomial-time approximation algorithms for \flex, with approximation ratios of $4, 3$ and $5$, respectively, given that \tecs and \tap admit polynomial-time 2-approximation algorithms. We defer the details to the later sections and instead present a road map for proving the main result.

Our approximate solution is obtained from returning the best of many solutions, each computed by one of the above three algorithms on an instance that is computed from the original instance by appropriately scaling the costs of the safe edges. The motivation for making safe edges cheaper is that buying a similarly priced unsafe edge instead likely incurs extra costs, since one safe edge or at least two edges have to cross each cut. The technical challenge is to determine the most useful scaling factors. 

The main idea in the analysis is to relate the costs of edges in an optimal solution to the costs of edges 
in the computed solutions based on a proper generalization of \emph{exchange bijections} between spanning trees. Exchange bijections $\varphi: A \rightarrow B$ are bijections between bases $A, B$ of a matroid (e.g., spanning trees in a connected graph), such that for any $a\in A$, the set $A \setminus \{a\} \cup \{\varphi(a)\}$ is again a basis of the matroid. It is well known that an exchange bijection always exists. We introduce our generalized notion of \emph{$\alpha$-monotone exchange bijections}, where $\alpha$ is a scaling factor used in the algorithm, and prove that they always exist between a spanning tree of the optimal and a spanning tree of a computed solution. We then combine the properties of such bijections with additional technical ideas to derive an upper bound on the cost of the computed solution. The bound is expressed in terms of several parameters that represent proportions of costs associated with parts of the computed and an unknown optimal solution, defined through the exchange bijections.

The final step is to combine all obtained upper bounds. Since we have the choice of selecting the scaling factors, but have no control over the parameters appearing in the upper bounds, we can compute a conservative upper bound on the approximation ratio by solving a three-stage \emph{factor-revealing min-max-min optimization problem}.
The inner minimum is taken over the upper bounds on the values of the solutions computed by algorithms \algoA, \algoB, and \algoC. The maximum is taken over the parameters that depend on an unknown optimal solution. Finally, the outer minimum is taken over the choice of scaling factors.
One interesting aspect of our factor revealing optimization problem is that its solution gives not only a  bound on the approximation ratio of the algorithm (as in, e.g.,~\cite{jain2003greedy,adjiashvili2014time}), but it also suggests optimal instance-independent scaling factors to be used by the algorithm itself.
One can show that by computing the scaling factors it is possible to constrain further
the unknown parameters in the problem, thus obtaining better instance-specific bounds.
While we can compute these factors in polynomial time, we will not elaborate on this approach in the paper.

Since the overall factor-revealing optimization problem is a three-stage
min-max-min program, we can only provide analytic proof for its optimal value
for very small sizes.
However, we are still able to give an analytic bound of \boundtwo in order to
prove Theorem \ref{thm:flex} by combining only algorithms \algoA and \algoC,
but using the optimal choices of scaling factors for a given instance.
To achieve the factor \boundthreehalf for bounded weight instances we use all
three algorithms \algoA, \algoB and \algoC to bound the optimal solution of the
min-max-min problem.  Clearly, using all three algorithms yields better bounds,
but we cannot give an analytic upper bound.  Instead we give a computational
upper bound on the min-max-min problem using the baron solver~\cite{baron}. An overview of
our approximation guarantees along with some related results can be found in
Table~\ref{tab:apx:overview}.

\subsection{Further Results}

We also consider several generalizations and special cases of \flex. First, we
show that the unweighted version of the problem admits a
$23/16$-approximation algorithm. We note that unweighted \flex does not
contain unweighted \tap as special case, and hence, our result does not imply a
$23/16$-approximation algorithm for unweighted \tap. In particular
we prove the following theorem for a generalization of \flex, which we call
$k$-\flex. The problem $k$-\flex asks for the minimum weight connected subgraph of a
given graph that can withstand the failure of at most $k$ \unsafe edges. Note
that $1$-\flex is simply \flex. 

\begin{restatable}{theorem}{unweightedkflex}
  \label{thm:unweightedkflex}
  Unweighted $k$-\flex admits a polynomial-time $\left( \vartheta_{k+1}\cdot \frac{2k+1}{2k+2} + \frac{1}{k+1}\right)$-approximation algorithm.
\end{restatable}

By $\vartheta_{k+1}$ we denote the approximation ratio of an approximation
algorithm for the minimum $(k+1)$-edge connected spanning subgraph problem. In
particular this implies a polynomial-time $23/16$-approximation algorithm for
unweighted \flex. Note that the approximation guarantee in
Theorem~\ref{thm:unweightedkflex} tends to one as $k$ tends to infinity. We
prove this result in Section~\ref{sec:unweighted}.

We contrast our algorithmic results by showing in
Section~\ref{sec:onerst:hardness} that a natural generalization of \flex to
matroid optimization is \NP-hard to approximate within any sublogarithmic
factor. 

\subsection{Additional Related Work}

Most classical network design problems are concerned with constructing a network with sufficient redundancy to deal with various kinds of reliability considerations. These include the $k$-edge connected spanning subgraph problem~\cite{gabow_et_al_kECSS_09,Gabow:12}, the survivable network design problem~\cite{Jain2001,Goemans:94}.

A field of optimization that deals directly with reliability issues is robust optimization (see e.g.~\cite{olver2010robust,adjiashvili2012structural} for thorough literature reviews on the topic). The Bulk-Robust framework was introduced by Adjiashvili, Stiller and Zenklusen~\cite{bulk} and later studied further in~\cite{adjiashvili_bindewald_michaels_icalp2016,hommelsheim2019secure,adjiashvili2015non,Bindewald:18}. Robustness with respect to the cost structure of the problem represents another important avenue of research. Many classical problem have been studied in this framework~\cite{kasperski2016robust,aissi2009min,kouvelis2013robust}.

\subsection{Notation}
\label{subsec:notation}

Unless stated otherwise graphs are loopless but may have parallel edges. Let $G
= (V, E)$ be a graph with vertex set $V$ and edge set $E$. We denote by $E(G)$
the edge set of $G$. For an edge $e$ we may write $G - e$ (resp., $G + e$) for
the graph $(V, E \setminus \{e\})$ (resp., $(V, E \cup \{e\}$). For an edge set
$E' \subseteq E$ we denote by $G / E'$ the graph obtained from $G$ by
contracting the edges in $E'$.  We denote by $\vartheta_k$ the ratio of an
approximation algorithm for the problem of finding a minimum-cardinality
$k$-edge connected spanning subgraph. Similarly, we denote by $\lambda$ (resp.,
$\tau$) the ratio of an approximation algorithm for \tecs (resp., \WTAP).
 
For the remainder of this paper we fix an instance $\I = (G, w, \overline{F})$ of
\onerst and some optimal solution $Z^* \subseteq E(G)$ of $\I$. 
To avoid technicalities, we add to each safe edge has a parallel \unsafe edge of the same cost.
It is readily seen that this modification preserves optimal solutions (a
solution that uses both edges of a parallel pair is not optimal and we may
always just pick the safe edge in an optimal solution). 
Observe that the solution $Z^*$
has the following structure. For some $r \in \nat$, the graph $(V(G), Z^*)$ consists of $r$ 2-edge-connected
components $C_1, C_2, \ldots, C_r$ that are joined together by safe edges $\sce
:= \{ \safe{f}_1, \safe{f}_2, \ldots, \safe{f}_{r-1} \} \subseteq \safe{F}$ in
a tree-like fashion. That is, if we contract each component $C_i$ to a single
vertex, the remaining graph is a tree $T^*$ with edge set $\sce$. On the other
hand, if we contract $\sce$, we obtain a 2-edge-connected spanning subgraph of
the resulting graph. We let $\delta := w(\sce) /\OPT(I)$. That is, the value
$\delta$ is the proportion of the cost of the safe cut edges $\sce$ relative to
the total cost of the optimal solution $Z^*$.


\section{Three Simple Approximation Algorithms for \onerst}
\label{sec:onerst:basicalgo}

We give three basic approximation algorithms for \flex that use in a
black-box fashion approximation algorithms for \tecs and \tap. 
With the known best algorithms for these problems, it is readily seen that
first one has an approximation guarantee of $(2+2\delta)$ and the second has a
guarantee of 3.
Using our technical tools from Section~\ref{sec:onerst:technical} one can show
that the third algorithm is a 5-approximation algorithm for \flex. Interestingly, it
performs much better precisely when the other two algorithms exhibit their
worst-case behavior. Note that none of the three algorithms attains the integrality gap of $8/3$ exhibited by the instance shown in Figure~\ref{fig:integrality-gap}.

\paragraph{Algorithm \algoA: a $(2+2\delta)$-approximation algorithm.}
Algorithm \algoA 
computes (in polynomial time) a $\lambda$-approximate 2-edge connected spanning
subgraph, e.g., by Jain's algorithm~\cite{Jain2001}. Algorithm \algoA then
removes all unsafe edges that are parallel to safe edges of spanning subgraph returns the resulting
edge set. Observe that the returned solution is feasible.
We now argue that Algorithm \algoA is a $(2+2\delta)$-approximation algorithm.
By adding a copy of each edge in $\sce$ to the optimal solution $Z^*$ we obtain a 2-edge-connected spanning subgraph $H_1^*$ of $G$.
The cost $w(H_1^*)$ of $H_1^*$ is given by
\[
  w(H_1^*) =  \sum_{e \in Z^* \setminus \sce} w(e) + 2 \cdot \sum_{e \in E(Z^*) \cap \sce} w(e)\enspace.
\]
and since $w(H_1) \leq 2 w(H_1^*)$, it follows that $w(H_1) \leq 2 w(H_1^*) = 2 (1 - \delta ) \cdot\OPT(I) + 4 \delta \OPT(\I)= (2 + 2 \delta) \cdot\OPT(\I)$
as claimed. Observe that Algorithm \algoA performs best if the weight of the safe edges $\sce$ is small.

\paragraph{Algorithm \algoB: a 3-approximation algorithm.}
Algorithm \algoB  computes a minimum spanning tree $T$ of $G$ and then
computes a $\tau$-approximate solution to the \WTAP instance $(G, T', w)$,
where the tree $T'$ is obtained by contracting each safe edge of $T$. 
The solution of the \WTAP instance together with the
tree $T$ is a feasible solution for $\I$ by
Lemma~\ref{lemma:krst:structure}. Since $w(T) \leq \OPT(I)$, the best available algorithms for \tap (see Table~\ref{tab:apx:overview}) give a 
3-approximation for \flex, and a $5/2$-approximation
for \flex on bounded-weight instances.

\paragraph{Algorithm \algoC:}
Algorithm \algoC first computes a minimum spanning tree $T$ of $G$. Let $G'$ be
the graph obtained from $G$ by contracting the safe edges of $T$, that is, $G':= G/(E(T) \cap \safe{F})$. Algorithm \algoC then computes a $\lambda$-approximate 2-edge-connected spanning subgraph 
$H' \subseteq E(G')$
and returns the edge set $(E(T)
\cap \overline{F}) \cup H'$. It is readily seen that Algorithm \algoC computes a
feasible solution and that the safe edges of $T$ have cost at most $\OPT(I)$. Using
exchange bijections from Section~\ref{sec:onerst:technical}, it can be
shown that $w(H') \leq 4\OPT(I)$, so Algorithm \algoC is a 5-approximation
algorithm for \onerst.

\section{An Improved Approximation Algorithm}
\label{sec:onerst:improved}

In this section we describe our improved approximation algorithm for \onerst,
which hybridizes the three basic algorithms for \onerst discussed in
Section~\ref{sec:onerst:basicalgo}.
We first illustrate why Algorithm \algoB, which has the best approximation guarantee of the three simple algorithms, may perform poorly.
Recall that Algorithm \algoB first computes a minimum spanning tree of $G$ and then
approximates the \WTAP instance $(G, T', w)$, where the tree $T'$
is obtained from the minimum spanning tree by contracting its safe edges.
Consider a cut edge $e \in \safe{F}$ of an optimal solution to $\I$. 
If Algorithm \algoB chooses instead of $e$ a slightly cheaper \unsafe edge $f$ across
the same cut in the MST computation, then making this cut ``safe'' in the second step by
buying either the edge $e$ or another edge of similar cost across the
cut. However, we can only $\tau$-approximate this step. Hence, if
Algorithm \algoB chooses $f$ instead of $e$, it may result in cost of up to
$3w(e)$.
We try to avoid such a situation by scaling the weight of all safe edges by a
suitable factor $\alpha \in [0, 1]$, hence making safe edges more attractive.  

Algorithm~\ref{alg:onerst:approx},
our improved approximation algorithm for \flex, proceeds as follows. It first computes suitable
scaling factors $\Alphas \subseteq [0, 1]$ (called ``threshold values'') for the costs of the safe edges.  
Then,
Algorithm~\ref{alg:onerst:approx} runs Algorithm \algoA using the original weights $w$ 
to obtain
solution $Z^A$. 
We say that we run Algorithm \algoB (resp., \algoC) \emph{with
scaling factor $\alpha$} if the minimum spanning tree in Algorithm \algoB
(resp., \algoC) is computed with respect to weights obtained from $w$ by
scaling the costs of the safe edges by $\alpha$.
Algorithm~\ref{alg:onerst:approx} runs Algorithms \algoB and \algoC with each
scaling factor $\alpha \in \Alphas \cup \{0, 1\}$ and returns a solution of
minimal weight among all the different solutions computed by Algorithms \algoA,
\algoB, and~\algoC. 

\begin{algorithm}[t]
    \caption{Improved Approximation Algorithm for \onerst}
  \label{alg:onerst:approx}
  \SetStartEndCondition{ }{}{}%
  \SetKwFunction{Range}{range}
  \SetKw{KwTo}{to}%
  \SetKw{KwAnd}{and}%
  \SetKwFor{For}{for}{ do}{}%
  \SetKwIF{If}{ElseIf}{Else}{if}{ then}{else if}{else}{}%
  \SetKwFor{While}{while}{ do}{fintq}%
  \AlgoDontDisplayBlockMarkers\SetAlgoNoLine\SetAlgoNoEnd
  \DontPrintSemicolon
  \SetKwInOut{Input}{input}
  \SetKwInOut{InOut}{in/out}
  \SetKwInOut{Output}{output}
  \Input{Instance $I = (G, w, \safe{F})$ of \onerst}

    compute threshold values $\Alphas := \{\alpha_e \mid e \in E(G)\} \cup \{0, 1\}$\;
    run Algorithm \algoA on $\I$ to obtain solution $Z^A$\;
    \For{threshold value $\alpha \in \Alphas$}
    {
        run algorithms \algoB and \algoC with scaling factor $\alpha$ to obtain solutions $Z_\alpha^B$ and~$Z_{\alpha}^C$, respectively\;
    }
    \Return{solution with lowest cost among $Z^A$ and $\{ Z_\alpha^B, Z_\alpha^C \mid \alpha \in \Alphas \}$ }\;
\end{algorithm} 

From the discussion in Section~\ref{sec:onerst:basicalgo} it follows that
Algorithm~\ref{alg:onerst:approx} returns a feasible solution.  It runs in
polynomial time if there are polynomially many threshold values and if they can be
computed efficiently. 
We defer the proofs of the running-time to Section~\ref{sec:onerst:technical}.

Let us denote by $\algo(I)$ the weight of the solution returned by
Algorithm~\ref{alg:onerst:approx}. 
Using properties of the threshold values we show that the selection of the scaling factors in
Algorithm~\ref{alg:onerst:approx} is best possible.

\begin{lemma}
    Let $\Alphas' \subseteq [0, 1]$ and let $\algo'(I)$ be the weight of the solution
    returned by Algorithm~\ref{alg:onerst:approx} run with $\Alphas$ set to $\Alphas'$.
    Then $\algo(I) \leq \algo'(I)$.
    \label{lemma:onerst:bestalpha}
\end{lemma}

A detailed analysis of the approximation ratio of
Algorithm~\ref{alg:onerst:approx} is deferred to 
Section~\ref{sec:onerst:5/2}. Here, we give a
high-level overview. Our starting point is 
Lemma~\ref{lemma:onerst:bestalpha}, which allows us to assume that Algorithm~\ref{alg:onerst:approx} tries \emph{all} scaling factors in $[0, 1]$.
We show that the approximation ratio of
Algorithm~\ref{alg:onerst:approx} is bounded from above by the optimal value of a min-max-min optimization problem. 
For an instance $I$ of \onerst and some $N \in \nat$, the optimization problem has the following data.
\begin{itemize}
\item Scaling factors $\alpha_1, \alpha_2, \ldots, \alpha_N \in [0, 1]$.  Due
  to the discussion above, in our analysis of Algorithm~\ref{alg:onerst:approx}
  we are free to choose these values.
\item Parameters $\beta_1, \beta_2, \ldots, \beta_N, \gamma_1, \gamma_2,
  \ldots, \gamma_N, \delta \in [0, 1]$, which depend on the structure of an
  optimal solution. These parameters additionally satisfy $\sum_{j=1}^{N}
  \beta_j + \sum_{j=1}^{N} \gamma_j = 1$.
\item Functions $f^A$, as well as $f_1^B, f_2^B, \ldots, f_N^B$ and $f_1^C,
  f_2^C, \ldots, f_N^C$ that  bound from above in terms of
  $\alpha_i$, $\beta_i, \gamma_i$, $1 \leq i \leq N$, $\lambda$ and $\tau$, the cost of the solutions.
\end{itemize} 
Precise definitions of the parameters and the functions will be given in Section~\ref{sec:onerst:5/2}. 
We note that Algorithm~\algoB is only used in the proof of Theorem~\ref{thm:flexbounded}.
Technically, we show that for a proper choice of functions $f^\algoA$, $f_i^\algoB$ and $f_i^\algoC$, the ratio of Algorithm~\ref{alg:onerst:approx} is bounded by the optimal value of the following optimization problem.
\begin{equation}
    \begin{aligned}
          \min_{\alpha_i \in [0, 1] \,:\, 1 \leq i \leq N} \;\max_{\substack{\beta_i \in [0, 1] \,:\, 1 \leq i \leq N\\ \gamma_i \in [0, 1] \,:\, 1 \leq i \leq N}} \;\min_{1 \leq i \leq N} \qquad&  \{f^\algoA(\cdot),\, f_i^\algoB(\cdot),\, f_i^\algoC(\cdot) \}\\
          \text{subject to}\qquad& \sum_{j=1}^{N} \beta_j + \sum_{j=1}^{N} \gamma_j = 1
    \end{aligned}
    \label{eq:minmaxmin}
\end{equation}

Our goal in Section~\ref{sec:onerst:5/2}
is to provide suitable functions $f^\algoA$, 
and $f_i^\algoC$ for our analysis. In Section~\ref{sec:onerst:5/2} we will give an analytic upper
bound of \boundtwo on the optimal value of~\eqref{eq:minmaxmin}.

\section{$\alpha$-MSTs, Thresholds, and Exchange Bijections}
\label{sec:onerst:technical}

In this section we present our main technical tools that are needed for the analysis of
Algorithm~\ref{alg:onerst:approx}.

\subsection{$\alpha$-MSTs and Thresholds}
\label{sec:onerst:threshold}

We first show that safe edges and  unsafe edges exhibit a ``threshold''
behavior with respect to MSTs if the costs are scaled by some $\alpha \in [0, 1]$. Furthermore, we show that i) the corresponding
threshold values can be computed in polynomial time, which is essential to
ensure that Algorithm~\ref{alg:onerst:approx} runs in polynomial time and ii)
they are the best
choice of scaling factors for Algorithm~\ref{alg:onerst:approx}, which allows us to assume in our analysis that we
execute Algorithm~\ref{alg:onerst:approx} for \emph{all} scaling factors $\alpha \in [0,
1]$. 
For $\alpha \in [0, 1]$, we denote by 
\[
    w_\alpha(e) = 
    \begin{cases}
        \alpha \cdot w(e) & \text{if $e \in \overline{F}$, and}\\
        w(e)              & \text{otherwise}
    \end{cases}
\]
the weight function obtained from $w$ by scaling the costs of the safe edges by
$\alpha$.
A spanning tree $T$ is called \emph{$\alpha$-minimum spanning tree}
($\alpha$-MST) if $E(T)$ has minimal weight with respect to $w_\alpha$.  

Consider changing the
scaling factor $\alpha$ smoothly from $0$ to $1$. We observe that for any safe
edge $e$, if there is an $\alpha$-MST containing $e$, then there is also an
$\alpha'$-MST containing $e$ for any $\alpha' \leq \alpha$.  On the other hand,
if there is an $\alpha$-MST containing an \unsafe edge $f$ then there is also
an $\alpha'$-MST containing $f$ for any $\alpha \leq \alpha' \leq 1$.  We
formally capture this notion in the following definition.

\begin{definition}
    \label{def:alphaexchange}
    Let $e \in E$ and $\alpha_e \in [0, 1]$. We say that
    $\alpha_e$ is a \emph{lower threshold for $e$} if for any $\alpha \in [0,
    1]$ there is an $\alpha$-MST containing $e$ if and only if $\alpha \geq
    \alpha_e$. If $e$ is in no $\alpha$-MST for $0 \leq \alpha \leq 1$, we define the lower threshold value of $e$ to be $\infty$.
    Similarly, $\alpha_e$ is an \emph{upper threshold for $e$} if
    for $\alpha \in [0, 1]$ there is an $\alpha$-MST containing $e$ if and only
    if $\alpha \leq \alpha_e$. The threshold values of an instance $\I = (G, w, \overline{F})$ is defined as $\{ \alpha_e \mid e \in E(G) \}$.
\end{definition}

The following technical lemma ensures the existence of threshold values for safe and unsafe edges.

\begin{lemma}
    \label{lemma:onerst:alphae:exists}
    For each \unsafe edge $f \in F$ there is a lower threshold $\alpha_f \in [0, 1]
    \cup \{\infty\}$. 
    For each safe edge $e \in \safe{F}$ there is an upper threshold $\alpha_e \in [0, 1]$.
\end{lemma}
\begin{proof} 
    We first prove the following claim.
    \setcounter{myclaim}{0}
    \begin{myclaim}
        \label{claim:alphae:unsafe}
        Let $f \in F$ be an \unsafe edge. If $f$ is in some $\alpha$-MST then for any $\alpha'
        \geq \alpha$, there is some $\alpha'$-MST containing $f$.
    \end{myclaim}
    \begin{proof}
        Let $\alpha' \geq \alpha$ and let us fix some $\alpha'$-MST $T_{\alpha'}$. Suppose for a
        contradiction that $f$ is in $T_\alpha$ but not in $T_{\alpha'}$ and assume
        that $f$ has smallest weight among all such edges. Consider the edges $e_1,
        e_2, \ldots, e_m$ of $G$ ordered non-decreasingly by their weight
        $w_\alpha$ and suppose $f = e_i$.  Similarly, let $e'_1, e'_2, \ldots,
        e'_m$ be the edges of $G$ ordered non-decreasingly by $w_{\alpha'}$ and
        suppose that $f = e'_{i'}$.  Note that due to the construction of the
        weight function, the weights of all edges in $\safe{F}$ are scaled by the
        same factor $\alpha$ and the weights of the edges $F$ are the same for
        $w_\alpha$ and $w_{\alpha'}$. 
        Therefore, we have that $i' \leq i$ and that $\{e'_1, \ldots, e'_{i'} \} \subseteq \{e_1, \ldots, e_{i} \}$.

        For $1 \leq i \leq m$, let $T^i_\alpha$ (resp., $T^{i}_{\alpha'}$) be the
        restriction of $T_\alpha$ (resp., $T_{\alpha'}$) to $e_1, e_2, \ldots, e_i$
        (resp., $e'_1, e'_2, \ldots, e'_{i}$).
        Since $f = e_{i'}$ is not in $T_{\alpha'}^{i'}$, the graph
        $T_{\alpha'}^{i'} + f$ contains a unique cycle $C_f$. Let $S \coloneqq E(C_f)
        \setminus E(T_\alpha^i)$.  Since $f = e_i$ is in $T_\alpha^i$, the set $S$
        is non-empty. For each $e \in S$, the graph $T_\alpha^i + e$ contains a
        unique cycle $C_e$. Hence, the edge set
        \[
            C' \coloneqq (E(C_f) \setminus S) \cup \bigcup_{e \in S} E(C_e) - e
        \]
        contains a cycle, but $C' \subseteq E(T_\alpha^i)$, which contradicts our
        assumption that $T_\alpha$ is a tree.
    \end{proof}

    Let $f \in F$ be an unsafe edge. If $f$ is not contained in some $1$-MST,
    then it is not contained in any $\alpha$-MST for $0 \leq \alpha \leq 1$.
    Therefore, the edge $f$ has a lower threshold value $\alpha_f = \infty$.
    Consider the case that $f$ is contained in some $\alpha$-MST, where $0 \leq
    \alpha \leq 1$. We choose $\alpha_f$ to be the smallest value $\alpha \in
    [0, 1]$, such that there is an $\alpha$-MST containing $f$. By
    Claim~\ref{claim:alphae:unsafe}, we have that $\alpha_f$ is a lower
    threshold value for $f$.

    We now prove the existence of an upper threshold value for safe edges.
    The proof of the following claim is analogous Claim~\ref{claim:alphae:unsafe}.
    \begin{myclaim}
        \label{claim:alphae:safe}
	Let $e \in \safe{F}$ be a safe edge. If $e$ is an edge of some
	$\alpha$-MST then for any $\alpha' \leq \alpha$, there is some
	$\alpha'$-MST containing $e$.
    \end{myclaim}

    Let $e \in \safe{F}$ be a safe edge. Oberserve that for $\alpha = 0$, there
    is an $\alpha$-MST containing $e$.  We let $\alpha_e$ be the largest value
    of $\alpha \in [0, 1]$, such that there is an $\alpha$-MST containing $f$.
    By Claim~\ref{claim:alphae:safe}, we have that $\alpha_e$ is an upper
    threshold value for $e$.
\end{proof}

It is easily seen that there are $O(|V(G)|^2)$ threshold values.  This implies
in particular that Algorithm~\ref{alg:onerst:approx} runs in polynomial time.
In fact, according to the next proposition there are at most $|V(G)|-1$
different threshold values.

\begin{proposition}
    \label{prop:onerst:threshold:compute}
    For each safe edge $e \in \safe{F}$ (resp., \unsafe edge $f \in F$), the
    upper threshold $\alpha_e$ (resp., lower threshold $\alpha_f$) can be
    computed in polynomial time. Furthermore, there are at most $|V(G)|-1$
    threshold values.
\end{proposition}
\begin{proof}
    Let $T_1$ be a 1-MST of $G$ and let $F_1 \coloneqq \{ f_1, f_2, \ldots, f_\ell \} =
    F \cap E(T_1)$ be the \unsafe edges of $T_1$.
    For each $e \in \safe{F}$ and $f_i \in F_1$, we initialize $\alpha^e_i \coloneqq w(f_i) / w(e)$.
    Thus, for $\alpha = \alpha_i^e$ we have $w_\alpha(e) = w_\alpha(f_i)$.
    We keep a set $\Alphas$ of all such $\alpha_i^e$ (actually triples $(e, i, \alpha_i^e)$).
    Since $|\safe{F}| \leq |E|$ and $|E(T_1)| < |V|$, we have that $\Alphas$
    has cardinality at most $|E||V|$.
    Additionally, we define a mapping 
    $\varrho : F_1 \rightarrow \safe{F} \setminus E(T_1)$. Whenever for some
    $\alpha \in [0, 1]$, an edge $f \in F_1$ is replaced in some $\alpha$-MST
    by a safe edge $\safe{f} \in \safe{F} \setminus E(T_1)$, we let $\varrho(f)
    \coloneqq \safe{f}$. 
    
    We perform at most $|V||E|$ iterations and in each iteration $j$, we do the
    following. We pick some $\alpha_i^e \in \Alphas$ of largest value and check
    whether the safe edge $e$ is in some $\alpha_i^e$-MST $T_{\alpha_i^e}$. If
    this is not the case, we remove $\alpha_i^e$ from $\Alphas$ and continue
    with the next iteration. Otherwise, we let $\varrho(f_i) \coloneqq e$ and
    $\alpha_{f_i} \coloneqq \alpha_i^e$ and for each $1 \leq i \leq \ell$, we delete
    $\alpha_i^e$ from $\Alphas$ and for each $e' \in \safe{F} \setminus E(T_1)$ we
    delete $\alpha_i^{e'}$ from $\Alphas$. Note that we can distinguish between the
    two cases in polynomial time by computing minimum-weight spanning trees.

    Observe that after the above algorithm terminates, the mapping $\varrho$
    assigns to each safe edge $e \in \safe{F} \setminus E(T_1)$ at most one
    partner $f \in F_1$, and to each \unsafe edge $f_i \in F_1$ at most one
    partner $e' \in \safe{F} \setminus E(T_1)$. 
    For each $e \in \safe{F} \cap E(T_1)$, we let $\alpha_e \coloneqq 1$ and for each
    $f \in F \setminus E(T_1)$, we let $\alpha_f \coloneqq \infty$. 
    Now, for each $f_i \in F_1 = F \cap E(T_1)$ such that $\varrho(f_i) = e$, we let $\alpha_e = \alpha_{f_i} \coloneqq \alpha_i^e$.
    Finally, for each $e \in \safe{F} \setminus T_1$ that is not in the image of
    $\varrho$ (resp., each $f \in F \cap E(T_1)$ that is not in the domain of
    $\varrho$), we let $\alpha_e \coloneqq 0$ (resp., $\alpha_f \coloneqq 0$).
    It is readlily verified
    that these choices are in accordance with the definition of lower and upper threshold values
    (Definition \ref{def:alphaexchange}). Finally, since there are at most $|V(G)|-1$ many
    \unsafe edges in $T_1$, there are at most $|V(G)|-1$  many threshold values.
    This concludes the proof.
\end{proof}

We now prove that threshold values are optimal scaling factors for
Algorithm~\ref{alg:onerst:approx}, as claimed in
Lemma~\ref{lemma:onerst:bestalpha}. 
\begin{proof}[Proof of Lemma~\ref{lemma:onerst:bestalpha}]
  Let $\alpha \in \Alphas'$ and let $\alpha_L$ (resp. $\alpha_R$) be the largest
  (resp., smallest) item in $\Alphas$, such that $\alpha_L \leq \alpha$ (resp.,
  $\alpha_R \geq \alpha$). Then, since $\Alphas$ contains a threshold value for each
  edge of $G$ and by the properties of the threshold values given in
  Definition~\ref{def:alphaexchange}, the tree $T_\alpha$ is either an
  $\alpha_L$-MST or an $\alpha_R$-MST of $G$. Therefore,
  Algorithm~\ref{alg:onerst:approx} has computed an $\alpha$-MST and a
  corresponding augmentation for each $\alpha \in \Alphas'$. We conclude that
  $\algo(\I) \leq \algo'(\I)$.
\end{proof}

\subsection{Exchange Bijections}
\label{sec:onerst:exchange}

In our analysis of Algorithm~\ref{alg:onerst:approx}, we will
use a charging argument based on the notion of monotone exchange bijections,
which we now introduce.
Let $G$ be a connected graph and let $T$ and $T'$ be spanning trees of $G$. A bijection $\varphi : E(T')
\to E(T)$ is called \emph{exchange bijection}, if for each $e \in
E(T')$, the graph $T' - e + \varphi(e)$ is a spanning tree of $G$.  
An exchange bijection $\varphi$ is \emph{monotone}, if for each edge $e \in
E(T')$ we have $w(e) \leq w(\varphi(e))$.
For any two spanning trees $T$ and $T'$ a canonical exchange bijection exists:
Note that the edge sets of spanning trees of $G$ are the bases of the graphic
matroid $M(G)$ of $G$. By the strong basis exchange property of matroids there
is a bijection between $E(T) \setminus E(T')$ and $E(T') \setminus
E(T)$ with the required properties, which can be extended to an exchange
bijection by mapping each item in $E(T) \cap E(T')$ to itself.
Furthermore, if $T'$ is an MST then for any spanning tree $T'$, a canonical
exchange bijection is monotone.

We generalize monotone exchange bijections as follows. 
\begin{definition}
    \label{def:alphamonotone}
    Let $\alpha \in [0, 1]$ and let $T$, $T'$ be spanning trees of $G$. An
    exchange bijection $\varphi: E(T') \to E(T)$ is 
    \emph{$\alpha$-monotone} if for each edge $e \in E(T')$ we have
    \begin{enumerate}
        \item $w(e) \leq \frac{1}{\alpha} w(\varphi(e))$, if $e \in \safe{F}$ and $\varphi(e) \in F$, and
        \item $w(e) \leq w(\varphi(e))$, if either $e, \varphi(e) \in \safe{F}$ or $e, \varphi(e) \in F$, and
        \item $w(e) \leq \alpha w(\varphi(e))$, if $e \in F$ and $\varphi(e) \in \safe{F}$.
    \end{enumerate}
\end{definition}

For any spanning tree $T$ of $G$, there is an
$\alpha$-monotone exchange bijection from an $\alpha$-MST to $T$.

\begin{lemma}
  \label{lem:onerst:existence:alpha:mon:exchange:bijection}
  Let $\alpha \in [0, 1]$, let $T_\alpha$ be an $\alpha$-MST of $G$ and let $T$ be any spanning tree of $G$.
  Then there is an $\alpha$-monotone exchange bijection $\varphi: E(T_\alpha) \to E(T)$.
\end{lemma}
\begin{proof}
  By the discussion above we have that there is a monotone exchange bijection
  $\varphi: T_\alpha \rightarrow T$ between an $\alpha$-MST $T_\alpha$ and any
  spanning tree $T$ of $G$ with respect to the weight function $w_\alpha$.  By
  substituting $w_\alpha$ with $w$ we observe that $\varphi$ is
  $\alpha$-monotone with respect to $w$.
\end{proof}

The following technical lemma is key to our charging argument in the analysis
of Algorithm~\ref{alg:onerst:approx} in sections~\ref{sec:onerst:5/2}
and~\ref{sec:onerst:bounded}.

\begin{lemma}
    \label{lemma:onerst:charging}
    Let $\alpha, \alpha' \in [0, 1]$, let $T$ be a spanning tree contained in an optimal solution to $\I$, and let
    $T_\alpha$ (resp., $T_{\alpha'}$) be an $\alpha$-MST (resp.,
    $\alpha'$-MST) of $G$. Then, for an $\alpha$-monotone exchange bijection $\varphi :
    E(T_\alpha) \to E(T)$ there is an $\alpha'$-monotone exchange bijection
    $\varphi' : E(T_{\alpha'}) \to E(T)$, such that $\varphi(e) =\varphi'(e)$
    for each $e \in E(T_\alpha) \cap E(T_{\alpha'})$.
\end{lemma}
\begin{proof} 
    Without loss of generality, let $\alpha \leq \alpha'$ and let $\varphi :
    E(T_\alpha) \to E(T)$ be an $\alpha$-monotone exchange bijection.  Let $0
    \leq q_1 < q_2 < \ldots < q_n \leq 1$ be the threshold values for $E(G)$
    with respect to the weights $w$. By the definition of threshold values, for
    each threshold value $q$, there are at least two $q$-MSTs.  Furthermore, we
    may assume without loss of generality that for each threshold value $q$,
    there are precisely two $q$-MSTs.  If this is not the case, then the reason
    is that there are at least two different pairs of edges, such that the two
    edges of each pair have the same scaled cost. We may break ties in an
    arbitrary but consistent way by slightly perturbing the weights and hence
    obtain two different thresholds, one for each pair.  By iterating this
    argument, we have the claimed property that there are at most two $q$-MSTs
    for a threshold value $q$ and furthemore, that $T_\alpha$ (resp.,
    $T_{\alpha'}$) is an $\alpha$-MST (resp., $\alpha'$-MST) and $\varphi$ is
    an $\alpha$-monotone exchange bijection. 

    Observe that $T_\alpha$ is a $q_i$-MST, where $q_i$ is the smallest
    threshold value such that $\alpha \leq q_i$. Similarly, the tree
    $T_{\alpha'}$ is a $q_j$-MST, where $q_j$ is the largest threshold value
    such that $q_j \leq \alpha'$. 
    We will reduce the task of constructing an $\alpha'$-monotone exchange
    bijection $\varphi' : E(T_{\alpha'}) \to E(T)$ with the desired properties
    to the case that the symmetric difference of $T_\alpha$ and $T_{\alpha'}$
    has size at most two.
    If $q_i = q_j$ then this is the case. If not, 
    then note that by our assumption above there is precisely one $q_{i+1}$-MST
    $T_{i+1}$ that is also a $q_{i}$-MST. Furthermore, the size of the
    symmetric difference of $T_{i+1}$ and $T_{\alpha}$ is at most two. 
    We construct a $q_{i+1}$-monotone exchange bijection $\varphi_{i+1} :
    E(T_{i+1}) \to E(T)$ that agrees with $\varphi$ on $E(T_\alpha) \cap
    E(T_{i+1})$. We then replace $T_\alpha$ by $T_{i+1}$ and $\varphi$ by
    $\varphi_{i+1}$ and iterate our argument. In each step, we reduce the size
    of the symmetric difference with $T_{\alpha'}$ by two.

    We now show how to construct an exchange bijection $\varphi_2$ with the desired
    properties, given two $q_i$-MSTs $T_1$ and $T_2$ and a $q_i$-monotone
    exchange bijection $\varphi_1 : E(T_1) \to E(T)$, such that the symmetric
    difference of $E(T_1)$ and $E(T_2)$ has size exactly two.
    Note that if there is a threshold value
    $q_{i+1}$ then, without loss of generality, the tree $T_2$ is also a
    $q_{i+1}$-MST.
    Let $E(T_1) \setminus E(T_2) = \{ e\}$ and $E(T_2) \setminus E(T_1) = \{ e' \}$.
    Consider the bijection $\varphi_2: E(T_2) \rightarrow E(T)$ such that
    $\varphi_2(f) \coloneqq \varphi_1(f)$ for each $f \in E(T_1) \cap E(T_2)$ and
    $\varphi_2(e') \coloneqq \varphi_1(e)$ otherwise. 
    We first show that $\varphi_2$ is an exchange bijection. By the definition
    of $\varphi_2$ it suffices to consider the edge $e'$.  Since $e \notin
    E(T_2)$ but $e \in E(T_1)$ we have that $e'$ and $e$ are contained in a
    cycle of $T_2 + e$. Since $\varphi_1$ is an exchange bijection, the edge
    $\varphi_1(e)$ is on a cycle of $T_1 + \varphi_1(e)$. Therefore, the graph $T_2
    + \varphi_1(e) = T_{1} - e + e' + \varphi_1(e)$ contains a cycle visiting $e'$.
    We conclude that $\varphi_2$ is an exchange bijection.

    It remains to show that $\varphi_2$ is $q_{i}$-monotone.  
    Since $e'$ is contained in the $q_{i}$-MST $T_2$, we have that
    $w_{q_{i}} (e') \leq w_{q_{i}} (\varphi_1 (e))$.  Therefore $\varphi_2$
    is a $q_{i}$-monotone exchange bijection such that $\varphi_1$ and
    $\varphi_2$ agree on each edge in $E(T_i) \cap E(T_{i+1})$.
    Furthermore, if there is a threshold value $q_{i+1}$, then the exchange
    bijection $\varphi_2$ is also $q_{i+1}$-monotone.
\end{proof}

It will be more convenient for us to apply the following corollary rather than
Lemma \ref{lemma:onerst:charging}.

\begin{corollary}
  \label{cor:onerst:charging}
  Let $0 \leq \alpha_1 \leq \alpha_2 \leq \ldots \leq \alpha_N \leq 1$ and
  for $1 \leq i \leq N$ let 
  $T_i$ be an $\alpha_i$-MST. Furthermore let $T$ be a spanning tree in an optimal solution to $\I$ and let 
  $\varphi_1: E(T_1) \rightarrow E(T)$ be an $\alpha_1$-monotone exchange bijection.
  Then for $1 < i \leq N$ there are $\alpha_i$-monotone exchange bijections 
  $\varphi_i: E(T_i) \rightarrow E(T)$, such that $\varphi_{i-1}(e) = \varphi_i(e)$
  for each $e \in E(T_{i-1}) \cap E(T_i)$.
\end{corollary}
\begin{proof}
  For $1 < i \leq N$ inductively apply Lemma \ref{lemma:onerst:charging} in order to obtain 
  $\alpha_i$-monotone exchange bijections $\varphi_i: E(T_i) \rightarrow E(T)$ with the desired properties.
\end{proof}

\section{Algorithm~\ref{alg:onerst:approx} gives a \boundtwo-approximation}
\label{sec:onerst:5/2}

In this section we give an analytic upper bound of $\boundtwo$ on the approximation
ratio of Algorithm~\ref{alg:onerst:approx}. 
For our analysis it suffices to run Algorithm \algoA together with 
Algorithm \algoC.
That is, using $\alpha$-monotone exchange
bijections from Section~\ref{sec:onerst:technical}, we determine functions $f^\algoA(\cdot)$ and $f^\algoC(\cdot)$
for the optimization problem~\eqref{eq:minmaxmin}, where $f^\algoC(\cdot)$ depends on a selection of
scaling factors and some other parameters to be introduced shortly.
We then transform problem~\eqref{eq:minmaxmin} into a maximization problem which we solve analytically.
Recall that according to Lemma~\ref{lemma:onerst:bestalpha}, 
the selection of scaling factors in Algorithm~\ref{alg:onerst:approx} is
optimal. Surprisingly, a worst-case instance for our bounds $f^\algoA(\cdot)$
and $f^\algoC(\cdot)$ in fact has a single threshold value which is $1/\lambda$. 
However, to obtain the approximation ratio of $\boundtwo$ it is crucial
to execute Algorithm~\ref{alg:onerst:approx} with all threshold values of the given instance.

Let $\mathcal{I}(N)$ be a class of instances of \flex with at most $N$
threshold values in the sense of Definition~\ref{def:alphaexchange}.  In the
following, suppose that $I \in \mathcal{I}(N)$ and recall that an optimal
solution $Z^* \subseteq E(G)$ of $\I$ consists of $r$ 2-edge-connected
components $C_1, C_2, \ldots, C_r$ that are joined together by safe edges $\sce
\coloneqq \{ \safe{f}_1, \safe{f}_2, \ldots, \safe{f}_{r-1} \} \subseteq \safe{F}$ in
a tree-like fashion. Moreover, for any spanning tree $T$ contained in the
optimal solution $Z^*$ we have $\sce \subseteq T$.

Observe that since there is an \unsafe edge for each safe edge of same weight
in $G$, we have that each threshold value of the safe edges is in $[0, 1]$. 
Let $0 \leq \alpha_1
\leq \alpha_2 \leq \ldots \leq \alpha_{N} \leq 1$ be the $N$ threshold values
of $I$ in non-decreasing order. 
In order to prepare our analysis, we consider 
for $i \in \{1, 2, \ldots, N\}$  an $\alpha_i$-MST $T_i$,  an $\alpha_i$-monotone exchange bijection $\varphi_i: T_i
\rightarrow T$ and a weight $w_i \coloneqq w_{\alpha_i}$.
For $2 \leq i \leq N$ we choose $\varphi_i$ such that for each $e \in E(T_{i-1}) \cap E(T_i)$ we have $\varphi_{i-1}(e) = \varphi_i(e)$  (in accordance with 
Corollary \ref{cor:onerst:charging}).
In order to define the parameters of the optimization problem~\eqref{eq:minmaxmin},
for $1 \leq i \leq N$, we partition
the edge set of the $\alpha_i$-MST $T_i$ into four parts $D_i$, $O_i$, $F_i$, and $S_i$ as
follows.
\begin{itemize}
    \item $D_i \coloneqq \{ e \in E(T_i) \cap F \mid \varphi_i(e) \in E' \}$
    \item $O_i \coloneqq \{ e \in E(T_i) \cap \safe{F} \mid \varphi_i(e) \in E' \}$
    \item $F_i \coloneqq \{ e \in E(T_i) \cap F \mid \varphi_i(e) \in E(T) \setminus E' \}$
    \item $S_i \coloneqq \{ e \in E(T_i) \cap \safe{F} \mid \varphi_i(e) \in E(T) \setminus E' \}$
\end{itemize}

The parameters of problem~\eqref{eq:minmaxmin} are given as follows.
For $1 \leq i \leq N$ we let $E^{\bar{F}}_i$ (resp., $E^F_i$) be the set of edges in $\sce$ (resp., $E(T) - \sce$) that have threshold value $\alpha_i$.
That is, $E^{\bar{F}}_i \coloneqq \{ e \in \sce \mid \alpha_e = \alpha_i \}$ and $E^F_i \coloneqq \{ e \in E(T) - \sce \mid \alpha_e = \alpha_i \}$.
For $1 \leq i \leq N$ we let $\beta_i = w(E^{\bar{F}}_i)/{\OPT(I)}$ and $\gamma_i =  w(E^F_i) /{\OPT(I)}$  be the fraction of the weight of the optimal solution
that is contributed by the edges in $E^{\bar{F}}_i$ (resp., $E^F_i$).
Finally, let $\xi \in [0, 1]$ be the the fraction of the weight of the optimal solution that is not contributed by the tree $T$; e.g., $\xi \coloneqq \frac{w(Z^*) - w(T)}{\OPT(I)}$.
The following properties of $\beta_i$, $\gamma_i$, $1 \leq i \leq N$, are readily verified: 
\begin{enumerate}
  \item $\beta_1, \beta_2, \ldots \beta_{N}, \gamma_1, \gamma_2, \ldots \gamma_{N}, \xi \in [0, 1]$, \label{prop:param:1}
  \item $\sum_{j=1}^{N} \beta_j = \frac{w(\sce)}{\OPT(I)}$, \label{prop:param:2}
  \item $\sum_{j=1}^{N} \gamma_j  = \frac{w(T - \sce)}{\OPT(I)}$, and \label{prop:param:3}
  \item $\xi = 1 - \sum_{j = 1}^{N} \beta_j - \sum_{j=1}^{N} \gamma_j$. \label{prop:param:4}
\end{enumerate}
We now bound the cost of the solutions $Z_i^C$ and $Z^A$ returned by Algorithm \algoC (resp., Algorithm \algoA) in terms of the parameters.

\begin{lemma}
  \label{lemma:onerst:refinedbound:C*}
  Suppose we run Algorithm~\ref{alg:onerst:approx} with the optimal threshold values $\Alphas = \{ \alpha_i \}_{1 \leq i \leq N}$.
  Let $Z_i^C$ be the solution computed by Algorithm \algoC with scaling factor $\alpha_i$ in Algorithm~\ref{alg:onerst:approx}. Then
  \begin{align*}
    &w(Z_i^C) \leq \left( 1 + \sum_{j = 1}^{i -1} (\lambda + \lambda \alpha_{j} -1) \beta_j + (\lambda -1)\cdot \sum_{j=1}^{N} \gamma_j + \sum_{j=i}^{N} \frac{1}{\alpha_j} \gamma_j  + (\lambda -1)\cdot \xi\right) \cdot\OPT(I).
  \end{align*}
\end{lemma}

\begin{proof}
  Let $T_i^S$ be the safe edges of the tree $T_i$ and let $\varphi_i: E(T_i)
  \rightarrow E(T)$ be an $\alpha_i$-monotone exchange bijection, where $T$ is
  a spanning tree of the optimal solution $Z^*$ to the instance $(G, w,
  \safe{F})$.
 
  By contracting each edge of $T_i^S$ in $G$ we obtain the graph $G^S_i \coloneqq G /
  E(T_i^S)$. Algorithm \algoC computes a $\lambda$-approximate solution to the instance
  $(G^S_i)$ of \ECSS.
  \setcounter{myclaim}{0}
  \begin{myclaim}
    \label{claim:H':feasibility}
    The set
    \[
	    Y_i \coloneqq \left(\bigcup_{e \in \sce \setminus \varphi_i (E(T_i^S))} \{ e, \varphi_i^{-1}(e)\}\right) \cup \bigcup_{1 \leq j \leq r} E(C_j)
    \]
    of edges is a feasible solution to the \ECSS instance $(G^S_i)$.
  \end{myclaim}
  \begin{proof}
    Clearly $(V, Y_i)$ is a connected graph. 
    It remains to argue that each edge $e \in Y_i$ is contained in some cycle. This is
    certainly true for each edge $e$ of a component $C_j$, $1 \leq j \leq r$.
    It remains to show that the edges in $\bigcup_{e \in \sce \setminus
    \varphi_i (E(T_i^S))} \{ e,  \varphi_i^{-1}(e)\}$ are contained in some cycle
    of $G^S_i$. Since $\varphi_i$ is an exchange bijection, an edge $e \in \sce \setminus \varphi_i(E(T_i^S))$ and its preimage $\varphi^{-1}(e)$ are on a cycle in $E(T) \cup \{e \}$.
    Since $G^S_i$ is obtained from $G$ by contracting the edges of the safe forest $T_i^S$, the edges $e$ and $\varphi^{-1}(e)$ are also on a cycle in $G_i^S$.
  \end{proof}
  
  We now bound the cost of $Z_i^C$.  By Claim \ref{claim:H':feasibility} we
  have that $T_i^S \cup Y_i$ is a feasible solution to the \onerst instance
  $(G, w, \safe{F})$. 
  The algorithm then returns in polynomial time a
  solution $Z_i^C$ of cost at most $w(Z_i^C) \leq w(T_i^S) + \lambda w(Y_i)$.  We
  first bound the cost of each edge of $T_i^S$ as follows.

  \begin{myclaim}
    \label{claim:T_i^S:boundcost}
    Let $e \in E(T_i^S)$ and let $\alpha_e$ be its threshold value. Then we have
    \[
      w(e) \leq 
      \begin{cases}
	\frac{1}{\alpha_e} \cdot w( \varphi_i (e)) & \text{if $\varphi_i(e) \notin \sce$, and}\\
	w( \varphi_i(e))  & \text{otherwise.}\enspace
      \end{cases}
    \]
  \end{myclaim}
  \begin{proof}
    First suppose that $\varphi_i(e) \in \sce$. Since $\varphi_i(e) \in \safe{F}$ and 
    $\varphi_i$ is an $\alpha_i$-monotone exchange bijection it follows that $w(e) \leq w(\varphi_i(e))$.
    Now let $\varphi_i(e) \notin \sce$. 
    Since $\varphi_i$ is an $\alpha_i$-monotone exchange bijection and since the threshold values of
    $e$ is $\alpha_e$, we have that $w(e) \leq \frac{1}{\alpha_e} w(\varphi_i (e))$.
  \end{proof}    
    Observe that, if $e \in T_i^S$, we also have $e \in T_j^S$ for each $j \leq i$.
    Since the exchange bijections $\varphi_1, \varphi_2, \ldots, \varphi_{N}$
    are in accordance with Corollary~\ref{cor:onerst:charging}, we then have 
    $\varphi_j(e) = \varphi_i(e)$ for every $j \leq i$.
    Thus, according to Claim~\ref{claim:T_i^S:boundcost} we have
    \begin{align*}
      w(T_i^S) & = \sum_{e \in T_i^S} w(e) = \sum_{e \in T_i^S} \sum_{j: \alpha_e = \alpha_j} w(e)\\
      & = \sum_{e \in T_i^S : \varphi_i(e) \in E'} \sum_{j: \alpha_e = \alpha_j} w(e) + \sum_{e \in T_i^S : \varphi_i(e) \notin E'} \sum_{j: \alpha_e = \alpha_j} w(e)\\
      & \overset{\text{Claim \ref{claim:T_i^S:boundcost}}}{\leq} \sum_{e \in T_i^S : \varphi_i(e) \in E'} \sum_{j: \alpha_e = \alpha_j} w(\varphi_i(e)) + \sum_{e \in T_i^S : \varphi_i(e) \notin E'} \sum_{j: \alpha_e = \alpha_j} \frac{1}{\alpha_e} \cdot w(\varphi_i(e))\\
      & \leq \left(\sum_{j=i}^{N} \beta_j  +  \sum_{j=i}^{N} \frac{\gamma_j}{\alpha_j}\right) \cdot \OPT(I) \\
      & = \left( 1 - \sum_{j=1}^{i-1} \beta_j - \sum_{j=1}^{N} \gamma_j - \xi  + \sum_{j=i}^{N} \frac{\gamma_j}{\alpha_j}\right) \cdot \OPT(I),
    \end{align*}
    where the last equality holds due to $\xi + \sum_{j = 1}^N \beta_j + \sum_{j=1}^N \gamma_j = 1$.
    We additionally bound the cost of $Y_i$.
    \begin{myclaim}
      \label{claim:Y_i:boundcost}
      \begin{align*}
	w(Y_{i}) \leq \left( \sum_{j=1}^{i-1} (1 + \alpha_{j} ) \beta_j + \sum_{j=1}^N \gamma_j + \xi \right) \cdot\OPT(I)\enspace.
      \end{align*}
    \end{myclaim}

    \begin{proof}
      We need to bound the cost of
      \[
	      Y_i = \left(\bigcup_{e \in \sce \setminus \varphi_i (E(T_i^S))} \{e,  \varphi_i^{-1}(e)\}\right) \cup \bigcup_{1 \leq j \leq r} E(C_j).
      \]
      We first bound the cost of $\bigcup_{e \in \sce \setminus \varphi_i (E(T_i^S))} \{e,  \varphi_i^{-1}(e)\}$.
      According to the definition of $\beta_1, \beta_2, \ldots, \beta_N$ we can bound the cost of 
      $\{e \mid e \in \sce \setminus \varphi_i (E(T_i^S))\}$ by $\sum_{j = 1}^{i-1} \beta_j \cdot \OPT(I)$.
      Additionally for each $e \in \sce \setminus \varphi_i (E(T_i^S))$ we have $w(\varphi_i^{-1}(e)) \leq \alpha_e w(e)$.
      Thus we can bound the cost of $\{\varphi_i^{-1}(e) \mid e \in \sce \setminus \varphi_i (E(T_i^S))\}$
      by $\sum_{j = i}^{i-1} \alpha_{j} \beta_j \cdot \OPT(I)$.
      Finally we bound the cost of $\bigcup_{1 \leq j \leq r} E(C_j)$ by $(\xi + \sum_{j=1}^N \gamma_j )\cdot \OPT(I)$.
      Putting things together we obtain 
      \begin{align*}
	w(Y_{i}) \leq \left( \sum_{j=1}^{i-1} (1 + \alpha_{j} ) \beta_j + \sum_{j=1}^N \gamma_j + \xi \right) \cdot\OPT(I)\enspace.
      \end{align*}
    \end{proof}
    Finally, since the algorithm computes a $\lambda$-approximate solution to the 2-ECSS instance,
    we have
    \begin{align*}
      w(Z_i^C) & \leq w(T_i^S) + \lambda \cdot w(Y_i)\\
      & \leq \left( \left( 1 - \sum_{j=1}^{i-1} \beta_j - \sum_{j=1}^{N} \gamma_j - \xi  + \sum_{j=i}^{N} \frac{\gamma_j}{\alpha_j}\right) + \left( \sum_{j=1}^{i} (\lambda + \lambda \alpha_j ) \beta_j + \sum_{j=1}^N \lambda \gamma_j + \lambda \xi \right) \right) \cdot \OPT(I)\\
      & \leq \left( 1 + \sum_{j = 1}^{i -1} (\lambda + \lambda \alpha_{j} -1) \beta_j + (\lambda -1)\cdot \sum_{j=1}^{N} \gamma_j + \sum_{j=i}^{N} \frac{1}{\alpha_j} \gamma_j  + (\lambda -1)\cdot \xi\right) \cdot\OPT(I).
    \end{align*}
    which concludes the proof.
\end{proof}

\begin{lemma}
  \label{lemma:onerst:refinedbound:A*}
  Suppose we run Algorithm~\ref{alg:onerst:approx} with the optimal threshold values $\Alphas = \{ \alpha_i \}_{1 \leq i \leq N}$.
  Let $Z^A$ be the solution 
  computed by Algorithm \algoA with scaling factor $\alpha_i$ in Algorithm~\ref{alg:onerst:approx}. Then
  \begin{align*}
    w(Z^A) \leq \left( \lambda + \lambda \cdot \sum_{j = 1}^{N} 
     \alpha_{j} \beta_j \right) \cdot\OPT(I).
  \end{align*}
\end{lemma}

\begin{proof}
The algorithm computes a $2$-edge connected spanning subgraph. Recall that for each 
safe edge $e \in \overline{F}$ there is a parallel \unsafe copy $e'$ of same cost.
We construct a feasible solution $Y_A$ to the 2-ECSS instance to bound the cost of $Z^A$.
  \setcounter{myclaim}{0}
\begin{myclaim}
\label{claim:Y_A:feasibility}
$Y_A \coloneqq Z^* \cup \{\varphi^{-1}_N(e)) \mid e \in \sce\}$ is a feasible solution to the 2-ECSS instance of cost at most
\[
w(Y_A) \leq \left( 1 + \sum_{j = 1}^{N} 
    \alpha_{j} \beta_j \right) \cdot\OPT(I).
    \]
\end{myclaim}
\begin{proof}
We first show the feasibility.
Clearly $(V, Y_A)$ is connected since it contains $Z^*$. We now show that each $e \in \sce$ is 
contained in some cycle in $Y_A$. 
Since for each safe edge there is an \unsafe edge of the same cost, 
we can assume that $\varphi^{-1}_N(e) \neq e$.
Then, by the definition of $\varphi_N$, the edge $e$ and its preimage $\varphi_N^{-1}$ 
are contained in a cycle.

It remains to bound the cost of $Y_A$.
We partition $Y_A$ into two edge-disjoint sets
$Y_A = Z^* \cup X_1$,
where $X_1 = \{\varphi^{-1}_N(e) \mid e \in \sce\}$,
and bound the cost of each part individually.
Clearly we have $w(Z^*) = \OPT(I)$.

To bound the cost of $X_1$, observe that for some $\varphi_N^{-1}(e) \in X_1$ we have
$w(\varphi_N^{-1}(e)) \leq \alpha_e w(e)$.
Thus, we may bound the cost of $X_1$ by
\[
  w(X_1) = \sum_{\varphi_N^{-1}(e) \in X_1} w(e) = \sum_{\varphi_N^{-1}(e) \in X_1} \sum_{j: \alpha_e = \alpha_j} w(\varphi_N^{-1}(e)) \leq \sum_{\varphi_N^{-1}(e) \in X_1} \sum_{j: \alpha_e = \alpha_j} \alpha_e w(e) = \sum_{j=1}^N \alpha_j \beta_j \cdot \OPT(I).
\]
Combining both bounds we obtain
\begin{align*}
w(Y_A) & \leq w(Z^*) + w(X_1) \leq  \left( 1 +  \sum_{j = 1}^{N} 
    \alpha_{j} \beta_j \right) \cdot\OPT(I).
\end{align*}
\end{proof}
Since the algorithm computes a $\lambda$-approximation, we have that 
  \begin{align*}
    w(Z^A) & \leq \lambda w(Y_A) \leq \left( \lambda + \lambda \cdot \sum_{j = 1}^{N} 
     \alpha_{j} \beta_j \right) \cdot\OPT(I).
  \end{align*}
\end{proof}

With the bounds from lemmas~\ref{lemma:onerst:refinedbound:C*}
and~\ref{lemma:onerst:refinedbound:A*} we show in the next lemma that 
problem~\eqref{eq:minmaxmin} can be simplified to the following maximization problem.
\begin{equation}
    \begin{aligned}
\max \ & \lambda \cdot \left( 1 + \sum_{j=1}^{N} \alpha_j \hat{\beta_j} \right)\\
\text{subject to } & \sum_{j=1}^{N} \hat{\beta_j} \cdot (1 + \alpha_j (\lambda - 1 + \lambda \alpha_j) ) = 1,\\
& 0 \leq \alpha_1 \leq \alpha_2 \leq \ldots \leq \alpha_N \leq 1,\\
& \hat{\beta_j} \in [0, 1] \text{ for all } j \in \{1, \ldots, N \}.
    \end{aligned}
    \label{eq:max}
\end{equation}

\begin{theorem}
  \label{thm:onerst:eqmax}
  The approximation guarantee of Algorithm \ref{alg:onerst:approx} for
  instances with at most $N$ threshold values is upper bounded by the optimal
  value of problem~\eqref{eq:max}.
\end{theorem}
\begin{proof}
  With the bounds from lemmas~\ref{lemma:onerst:refinedbound:C*}
  and~\ref{lemma:onerst:refinedbound:A*} and
  properties~\ref{prop:param:1}--\ref{prop:param:4} above, for instances with
  at most $N$ threshold values, we can rewrite problem~\eqref{eq:minmaxmin} as
  the following max-min optimization problem.
  \begin{equation}
    \begin{aligned}
      \;\max_{\substack{\xi \in [0, 1]\\ \alpha_i \in [0, 1] \,:\, 1 \leq i \leq N\\ \beta_i \in [0, 1] \,:\, 1 \leq i \leq N\\ \gamma_i \in [0, 1] \,:\, 1 \leq i \leq N}} \;\min_{1 \leq i \leq N} \qquad&  \{w(Z^\algoA),\, w(Z^\algoC_i)\}\\
      \text{subject to}\qquad& \sum_{j=1}^{N} \beta_j + \gamma_j + \xi = 1,\\
      & \alpha_1 \leq \alpha_2 \leq \ldots \leq \alpha_N
    \end{aligned}
    \label{eq:maxmin}
  \end{equation}
  We obtain from~\eqref{eq:maxmin} a simple maximization problem as follows.
  Since $\lambda \geq 1$, each $\beta_j, \gamma_j, j \in [N]$ as well as $\xi$
  has a positive coefficient in the bounds of the lemmas
  \ref{lemma:onerst:refinedbound:C*} and \ref{lemma:onerst:refinedbound:A*}.
  Moreover, since we maximize over $\beta_1, \beta_2, \ldots, \beta_N$,
  $\gamma_1, \gamma_2, \ldots, \gamma_N$ and $\xi$ we may assume $\xi = 0$.  To
  see this, suppose we have an optimal choice of the variables where $\xi > 0$.
  Then, consider the following new variables.  Let $\beta_j' = \beta_j$ for $j
  \in [N], \gamma_j' = \gamma_j$ for $j \in [N-1]$, $\gamma_N' = \gamma_N + \xi$
  and $\xi' = 0$. Now observe that the value of the minimum over $w(Z^\algoA)$ and
  $w(Z^\algoC_i), i \in [N]$ is at least as large for the new variables as for
  the old ones.  Thus we can assume that $\xi = 0$ and have 
  \begin{equation*}
    \sum_{j=1}^{N} \beta_j + \sum_{j=1}^{N} \gamma_j = 1.
  \end{equation*}

  Hence, it follows that the approximation guarantee of Algorithm
  \ref{alg:onerst:approx} is upper bounded by the following optimization
  problem
  \begin{equation*}
    \begin{aligned}
      \;\max_{\substack{\alpha_i \in [0, 1] \,:\, 1 \leq i \leq N\\ \beta_i \in [0, 1] \,:\, 1 \leq i \leq N\\ \gamma_i \in [0, 1] \,:\, 1 \leq i \leq N}} \;\min_{1 \leq i \leq N} \qquad&  \{w(Z^\algoA),\, w(Z^\algoC_i)\}\\
      \text{subject to}\qquad& \sum_{j=1}^{N} \beta_j + \gamma_j= 1,\\
      & \alpha_1 \leq \alpha_2 \leq \ldots \leq \alpha_N,
    \end{aligned}
    \label{eq:maxmin2}
  \end{equation*}
  where
  \begin{align*}
    w(Z_i^C) =  1 +  \sum_{j = 1}^{i-1} (\lambda -1 + \lambda \alpha_{j}) \beta_j +  (\lambda -1) \cdot \sum_{j=1}^{N} \gamma_j + \sum_{j=i}^{N} \frac{\gamma_j}{\alpha_j} 
  \end{align*}
  for each $i \in [N]$ and 
  \begin{align*}
    w(Z^A) = \lambda + \lambda \cdot \sum_{j = 1}^{N} \alpha_{j} \beta_j \enspace.
  \end{align*}
  Let us assume that the optimal value for this optimization problem is $\rho
  \in [\lambda, 2 \lambda]$ and let $\beta_1^*, \ldots, \beta_N^*$ and
  $\gamma_1^*, \ldots, \gamma_N^*$ be the optimal values of the respective
  variables. Since Algorithm \algoC gives a 5-approximation only 
  we know that the minimum of the optimization problem above is equal to
  $w(Z^\algoA)$. Thus we have
  \[
    \rho = \lambda \cdot \left( 1 + \sum_{j=1}^{N} \alpha_j \beta_j^* \right) \enspace.
  \]
  Therefore we have that 
  \begin{align*}
    w(Z_N^\algoC) - w(Z^\algoA) = & \ 1 +  \sum_{j = 1}^{N-1} (\lambda -1 + \lambda \alpha_{j}) \beta_j^* +  (\lambda -1) \cdot \sum_{j=1}^{N} \gamma_j^* + \frac{\gamma_N^*}{\alpha_N} - (\lambda + \lambda \cdot \sum_{j = 1}^{N} 
    \alpha_{j} \beta_j^*) \\
    = & \ 1 +  \sum_{j = 1}^{N} (\lambda -1 + \lambda \alpha_{j}) \beta_j^* - (\lambda -1 + \lambda \alpha_N) \beta_N^* +  (\lambda -1) \cdot \sum_{j=1}^{N} \gamma_j^*  + \frac{\gamma_N^*}{\alpha_N} - \lambda - \lambda \cdot \sum_{j = 1}^{N} 
    \alpha_{j} \beta_j^*\\
    = & \ \frac{\gamma_N^*}{\alpha_N} - (\lambda - 1 + \lambda \alpha_N) \beta_N^*\\
    \geq & \ 0, 
  \end{align*}
  where the third equality follows from $\sum_{j=1}^{N} \beta_j + \gamma_j= 1$.
  In fact, we may assume that the last inequality is an equality. If not, then
  we can reduce $\gamma_N^*$ by some fraction and increase $\gamma_{N-1}$ by
  the same fraction.  This yields a feasible solution but may increase the
  weight of the maximum.  Recursively for $1 \leq i \leq N-1$ we obtain 
  \begin{align*}
    w(Z_i^\algoC) - w(Z^\algoA) & =   1 +  \sum_{j = 1}^{i-1} (\lambda -1 + \lambda \alpha_{j}) \beta_j +  (\lambda -1) \cdot \sum_{j=1}^{N} \gamma_j + \sum_{j=i}^{N} \frac{\gamma_j}{\alpha_j} - \lambda + \lambda \cdot \sum_{j = 1}^{N} 
    \alpha_{j} \beta_j\\
    & = \frac{\gamma_i^*}{\alpha_i} - (\lambda - 1 + \lambda \alpha_i) \beta_i^*\\
    & = 0
  \end{align*}
  Thus we obtain $\gamma^*_i = \alpha_i (\lambda - 1 + \lambda \alpha_i) \beta_i^*$ for $1 \leq i \leq N$.
  Substituting this for $\gamma_j^*$ in $\sum_{j=1}^N \beta_j^* + \gamma_j^* = 1$, we obtain 
  \[
    \sum_{j=1}^{N} (1 + \alpha_i (\lambda - 1 + \lambda \alpha_i)) \beta_i^* = 1\enspace.
  \]
  Hence, in summary the optimization problem~\eqref{eq:minmaxmin} with
  bounds $w(Z^\algoA)$ and $w(Z_i^\algoC)$ from lemmas~\ref{lemma:onerst:refinedbound:C*}
  and~\ref{lemma:onerst:refinedbound:A*} can be written as maximization problem~\eqref{eq:max}.
\end{proof}

Next, we give an analytic solution to problem~\eqref{eq:max} in terms of
$\lambda$ which yields a ratio of $(8 + 4 \sqrt{2})/(4 + \sqrt{2}) <
\boundtwo$ for Algorithm~\ref{alg:onerst:approx} for $\lambda = 2$, the known
best approximation guarantee for \tecs. This is the bound claimed in Theorem \ref{thm:flex}. 

\begin{theorem}
  \label{thm:onerst:approx}
  Algorithm \ref{alg:onerst:approx} has an approximation guarantee of $\frac{\lambda \cdot (\lambda + 2 \sqrt{\lambda})}{2\sqrt{\lambda} + \lambda -1}$.
\end{theorem}
\begin{proof}
  Consider an optimal solution to problem \eqref{eq:max} with optimal values
  $\alpha_1^* < \alpha_2^* < \ldots < \alpha_N^*$ and $\hat\beta_1^*, \hat\beta_2^*,
  \ldots, \hat\beta_N^*$.  It is easy to see that an optimal solution has only one
  $\hat\beta_k^* \neq 0$, $k \in [N]$. We then have 
  \[ 
    \hat\beta_k^* = \frac{1}{1 + \alpha_k (\lambda - 1 + \lambda \alpha_k)}\enspace,
  \]
  where  $\alpha_k \in [0, 1]$.
  Optimizing over $\alpha_k \in [0, 1]$ yields the optimal value $\alpha_k
  = 1/\sqrt{\lambda}$.  Thus we obtain $\hat\beta_k^* = \frac{
  \sqrt{\lambda}}{2 \sqrt{\lambda}+\lambda-1}$ and the optimal value for 
  problem~\eqref{eq:max} is $\frac{\lambda \cdot (\lambda + 2
  \sqrt{\lambda})}{2\sqrt{\lambda} + \lambda -1}$.  Now observe that the
  solution does not depend on the number $N$ of threshold values. Hence the
  bound holds for instances with any number of threshold values.
\end{proof}

\section{Improved Approximation for Bounded Weights}
\label{sec:onerst:bounded}

In this section we give a computational proof for an upper bound of
$\boundthreehalf$ on the approximation ratio of
Algorithm~\ref{alg:onerst:approx} for bounded weight instances. Similar to
Section~\ref{sec:onerst:5/2} we give an upper bound on the value of the
min-max-min optimization problem~\eqref{eq:minmaxmin}, which gives in turn an
upper bound on the approximation ratio of Algorithm~\ref{alg:onerst:approx}. In
contrast to our previous analysis, we also use Algorithm~\algoB in order to
exploit recent progress in the approximation of \tap on bounded-weight
instances. Hence we consider the simultaneous worst-case behavior of all three
algorithms. However, this prevents us from  giving an improved \emph{analytic}
upper bound on the approximation ratio of Algorithm~\ref{alg:onerst:approx} in
this setting. Instead, we give a computational upper bound on the approximation
ratio of Algorithm~\ref{alg:onerst:approx} which is obtained by a non-linear
programming (NLP) solver.  
Since we are not aware of a solver for min-max-min optimization problems, we
relax the first $\min$ in the problem~\eqref{eq:minmaxmin} by selecting
appropriate scaling factors and transform the relaxation into a quadratic
maximization problem which we solve computationally.
A proof that our claimed upper bound indeed holds is given by the
branch-and-bound tree of the solver. Nevertheless, we encourage people to independently verify the bound by solving~\eqref{optforbound}.
Finding an analytic proof of the improved approximation guarantee
is an interesting open problem.  

We start similar to Section \ref{sec:onerst:5/2}.
Let $N \in \nat$ and let $0 = \alpha_0 \leq \alpha_1 \leq \alpha_2 \leq \ldots
\leq \alpha_{N+1} = 1$ (not necessarily the threshold values of the given instance) 
be $N+2$ values in $[0,1]$ in non-decreasing order. 
For $i \in \{0, 1, \ldots, N+1\}$ let $T_i$ be an $\alpha_i$-MST, let $\varphi_i: T_i
\rightarrow T$ be an $\alpha_i$-monotone exchange bijection, where $T$ is a
spanning tree of $Z^*$ and let $w_i \coloneqq w_{\alpha_i}$.  For $1 \leq i \leq N+1$
we choose $\varphi_i$ in accordance with Corollary
\ref{cor:onerst:charging}, that is, we have $\varphi_{i-1}(e) = \varphi_i(e)$
for each $e \in E(T_{i-1}) \cap E(T_i)$.  This will be useful later on for our
charging argument.  For $0 \leq i \leq N+1$, we partition the edge set of
$T_i$ into four parts $D_i$, $O_i$, $F_i$, and $S_i$ as in
Section~\ref{sec:onerst:5/2}.
\begin{itemize}
    \item $D_i \coloneqq \{ e \in E(T_i) \cap F \mid \varphi_i(e) \in E' \}$
    \item $O_i \coloneqq \{ e \in E(T_i) \cap \safe{F} \mid \varphi_i(e) \in E' \}$
    \item $F_i \coloneqq \{ e \in E(T_i) \cap F \mid \varphi_i(e) \in E(T) \setminus E' \}$
    \item $S_i \coloneqq \{ e \in E(T_i) \cap \safe{F} \mid \varphi_i(e) \in E(T) \setminus E' \}$
\end{itemize}

Similar to Section \ref{sec:onerst:5/2} we define the parameters of problem~\eqref{eq:minmaxmin}.
Since we consider arbitrary values of $\alpha$ in our analysis instead of the 
threshold values of the given instance, these parameters are defined slightly different.
The key difference is that in Section \ref{sec:onerst:5/2} we defined the 
$\beta$ and $\gamma$ values according to the precise
threshold values that we also used in the analysis. 
Since we do not consider the threshold values
in the analysis of this section, instead we define the parameters according 
to the closest upper and lower $\alpha$ values to the threshold value.

More formally, for $0 \leq i \leq N$ we let $E^{\bar{F}}_i$ (resp., $E^F_i$) be the set of edges in $\sce$ (resp., $E(T) - \sce$) that have threshold value between $\alpha_i$ and $\alpha_{i+1}$.
That is, $E^{\bar{F}}_i \coloneqq \{ e \in \sce \mid \alpha_i < \alpha_e \leq \alpha_{i+1} \}$ and $E^F_i \coloneqq \{ e \in E(T) - \sce \mid \alpha_i < \alpha_e \leq \alpha_{i+1} \}$.
For $0 \leq i \leq N$ we let $\beta_i = w(E^{\bar{F}}_i)/{\OPT(I)}$ and $\gamma_i =  w(E^F_i) /{\OPT(I)}$  be the fraction of the weight of the optimal solution
that is contributed by the edges in $E^{\bar{F}}_i$ (resp., $E^F_i$).
Finally, let $\xi \in [0, 1]$ be the the fraction of the weight of the optimal solution that is not contributed by the tree $T$; e.g., $\xi \coloneqq \frac{w(Z^*) - w(T)}{\OPT(I)}$.
The following properties of $\beta_i$, $\gamma_i$, $0 \leq i \leq N$, are readily verified: 
\begin{enumerate}
  \item $\beta_1, \beta_2, \ldots \beta_{N}, \gamma_1, \gamma_2, \ldots \gamma_{N}, \xi \in [0, 1]$, \label{prop:param:1}
  \item $\sum_{j=0}^{N} \beta_j = \frac{w(\sce)}{\OPT(I)}$, \label{prop:param:2}
  \item $\sum_{j=0}^{N} \gamma_j  = \frac{w(T - \sce)}{\OPT(I)}$, and \label{prop:param:3}
  \item $\xi = 1 - \sum_{j = 0}^{N} \beta_j - \sum_{j=0}^{N} \gamma_j$. \label{prop:param:4}
\end{enumerate}

With all these parameters we can now bound the cost of the solutions computed
by algorithms \algoA, \algoB, and \algoC. We start with the solution $Z_i^C$
returned by Algorithm \algoC.

\begin{lemma}
  \label{lemma:onerst:refinedbound:C}
  Suppose we run Algorithm~\ref{alg:onerst:approx} with threshold values $\Alphas = \{ \alpha_i \}_{0 \leq i \leq N}$.
  For $1 \leq i \leq N$ let $Z_i^\algoC$ be the solution 
  computed by Algorithm \algoC with scaling factor $\alpha_i$ in Algorithm~\ref{alg:onerst:approx}. Then
  \begin{align*}
    &w(Z_i^C) \leq \left( 1 +  \sum_{j = 0}^{i-1} (\lambda -1 + \lambda \alpha_{j+1}) \beta_j +  (\lambda -1) \cdot \sum_{j=0}^{N} \gamma_j + \sum_{j=i}^{N} \frac{\gamma_j}{\alpha_j} + (\lambda -1) \cdot \xi  \right) \cdot \OPT(I).
  \end{align*}
\end{lemma}
\begin{proof}
  The proof is similar to the one of Lemma \ref{lemma:onerst:refinedbound:C*}. 
  Let $T_i^S$ be the safe edges of the tree $T_i$ and let $\varphi_i: E(T_i)
  \rightarrow E(T)$ be an $\alpha_i$-monotone exchange bijection, where $T$ is
  a spanning tree of the optimal solution $Z^*$ to the instance $(G, w,
  \safe{F})$.
 
  By contracting each edge of $T_i^S$ in $G$ we obtain the graph $G^S_i
  \coloneqq G / E(T_i^S)$. Algorithm \algoC then computes a
  $\lambda$-approximate solution to the instance $(G^S_i)$ of \ECSS. The next
  claim is identical to Claim~\ref{claim:H':feasibility} in the proof of
  Lemma~\ref{lemma:onerst:refinedbound:C*}.
  \setcounter{myclaim}{0}
  \begin{myclaim}
    \label{claim:H':feasibility2}
    The set
    \[
      Y_i \coloneqq \left(\bigcup_{e \in \sce \setminus \varphi_i (E(T_i^S))} \{ e, \varphi_i^{-1}(e)\}\right) \cup \bigcup_{1 \leq j \leq r} E(C_j)
    \]
    of edges is a feasible solution to the \ECSS instance $(G^S_i)$.
  \end{myclaim}
  
  We now bound the cost of $Z_i^C$. The algorithm then returns in polynomial
  time a solution $Z_i^C$ of cost at most $w(Z_i^C) \leq w(T_i^S) + \lambda
  w(Y_i)$.  The next claim is very similar to Claim~\ref{claim:T_i^S:boundcost}
  of Lemma \ref{lemma:onerst:refinedbound:C*}. The only difference is that we
  consider arbitrary values of $\alpha$ in our analysis instead of the
  threshold values of the given instance. Thus we need to bound the cost of
  $w(e)$ in terms of the largest $\alpha \in W$ such that $\alpha \leq
  \alpha_e$ and the rest of the proof is analogous. We bound the cost of each
  edge of $T_i^S$ as follows.  
  \begin{myclaim}
    \label{claim:T_i^S:boundcost2}
    Let $e \in E(T_i^S)$ and let $\alpha_e$ be its threshold value. Then we have
    \[
      w(e) \leq 
      \begin{cases}
	\frac{1}{\alpha_j} \cdot w( \varphi_i (e)) & \text{if $\varphi_i(e) \notin \sce$, and}\\
	w( \varphi_i(e))  & \text{otherwise,}\enspace
      \end{cases}
    \]
    where $\alpha_j$ is the largest $\alpha \in \Alphas$ satisfying $\alpha \leq \alpha_e$. 
  \end{myclaim}

According to Claim~\ref{claim:T_i^S:boundcost2} we have for $1 \leq i \leq N$ that
    \begin{align*}
      w(T_i^S) & = \sum_{e \in T_i^S} w(e) = \sum_{e \in T_i^S} \sum_{\text{largest } j: \alpha_e \geq \alpha_j} w(e)\\
      & = \sum_{e \in T_i^S : \varphi_i(e) \in E'} \sum_{\text{largest } j: \alpha_e \geq \alpha_j} w(e) + \sum_{e \in T_i^S : \varphi_i(e) \notin E'} \sum_{\text{largest } j: \alpha_e \geq \alpha_j} w(e)\\
      & \overset{\text{Claim \ref{claim:T_i^S:boundcost2}}}{\leq} \sum_{e \in T_i^S : \varphi_i(e) \in E'} \sum_{\text{largest } j: \alpha_e \geq \alpha_j} w(\varphi_i(e)) + \sum_{e \in T_i^S : \varphi_i(e) \notin E'} \sum_{\text{largest } j: \alpha_e \geq \alpha_j} \frac{1}{\alpha_j} \cdot w(\varphi_i(e))\\
      & \leq \left(\sum_{j=i}^{N} \beta_j  +  \sum_{j=i}^{N} \frac{\gamma_j}{\alpha_j}\right) \cdot \OPT(I) \\
      & = \left( 1 - \sum_{j=0}^{i-1} \beta_j - \sum_{j=0}^{N} \gamma_j - \xi  + \sum_{j=i}^{N} \frac{\gamma_j}{\alpha_j}\right) \cdot \OPT(I)\enspace,
    \end{align*}
    where the last equality holds due to $\xi + \sum_{j = 1}^N \beta_j + \sum_{j=1}^N \gamma_j = 1$.
    Furthermore, we bound the cost of $Y_i$ as follows.
    \begin{myclaim}
	    \label{claim:Y_i:boundcost2}
	    \begin{align*}
		    w(Y_{i}) \leq \left( \sum_{j=0}^{i-1} (1 + \alpha_{j+1} ) \beta_j + \sum_{j=0}^N \gamma_j + \xi \right) \cdot\OPT(I)\enspace.
	    \end{align*}
    \end{myclaim}
    The proof of Claim~\ref{claim:Y_i:boundcost2} is analogous to the proof of
    Claim~\ref{claim:Y_i:boundcost} in
    Lemma~\ref{lemma:onerst:refinedbound:C*}. However, for each $e \in \sce
    \setminus \varphi_i (E(T_i^S))$, we have $w(\varphi_i^{-1}(e)) \leq
    \alpha_{j+1} w(e)$.  Here $\alpha_{j}$ is the largest $\alpha_j \in
    \Alphas$ satisfying $\alpha_j \leq \alpha_e$.  Thus we can bound the cost
    of $\{ \varphi_i^{-1}(e) \mid e \in \sce \setminus \varphi_i (E(T_i^S))\}$
    by $\sum_{j = i}^{i-1} \alpha_{j+1} \beta_j$.

    Finally, since the algorithm computes a $\lambda$-approximate solution to the 2-ECSS instance,
    we have
    \begin{align*}
	    w(Z_i^C) & \leq w(T_i^S) + \lambda \cdot w(Y_i)\\
	    & \leq \left( \left( 1 - \sum_{j=0}^{i-1} \beta_j - \sum_{j=0}^{N} \gamma_j - \xi  + \sum_{j=i}^{N} \frac{\gamma_j}{\alpha_j}\right) + \left( \sum_{j=1}^{i} (\lambda + \lambda \alpha_{j+1} ) \beta_j + \sum_{j=0}^N \lambda \gamma_j + \lambda \xi \right) \right) \cdot \OPT(I)\\
	    & \leq \left( 1 + \sum_{j = 0}^{i -1} (\lambda + \lambda \alpha_{j+1} -1) \beta_j + (\lambda -1)\cdot \sum_{j=0}^{N} \gamma_j + \sum_{j=i}^{N} \frac{1}{\alpha_j} \gamma_j  + (\lambda -1)\cdot \xi\right) \cdot\OPT(I)\enspace.
    \end{align*}
    which concludes the proof.
\end{proof}

Next, we bound the cost of the solution $Z^\algoA$ computed by Algorithm \algoA.

\begin{lemma}
\label{lemma:onerst:refinedbound:A}
  Suppose we run Algorithm~\ref{alg:onerst:approx} with threshold values
  $\Alphas = \{ \alpha_i \}_{0 \leq i \leq N}$.  Let $Z^\algoA$ be the solution
  computed by Algorithm \algoA in Algorithm~\ref{alg:onerst:approx}. Then
  \begin{align*}
    w(Z^A) \leq \left( \lambda + \lambda \cdot \sum_{j = 0}^{N} 
     \alpha_{j+1} \beta_j \right) \cdot\OPT(I)\enspace.
  \end{align*}
\end{lemma}

\begin{proof}
	The proof is similar to the one of Lemma \ref{lemma:onerst:refinedbound:A*}. 
	The only difference is an index shift in the bound of the following claim, compared to Claim~\ref{claim:Y_A:feasibility}
	of Lemma \ref{lemma:onerst:refinedbound:A*}.
	\setcounter{myclaim}{0}
	\begin{myclaim}
		\label{claim:Y_A:feasibility2}
		$Y_A \coloneqq Z^* \cup \{\varphi^{-1}_N(e)) \mid e \in \sce\}$
		is a feasible solution to the 2-ECSS instance of cost at most
		\[
			w(Y_A) \leq \left( 1 + \sum_{j = 0}^{N} 
			\alpha_{j+1} \beta_j \right) \cdot\OPT(I).
		\]
	\end{myclaim}
	\begin{proof}
	The proof for feasibility can be found in Claim~\ref{claim:Y_A:feasibility}
	of Lemma \ref{lemma:onerst:refinedbound:A*}.

	It remains to bound the cost of $Y_A$.
	We partition $Y_A$ into two edge-disjoint sets
	\[
		Y_A = Z^* \cup X_1,
	\]
	where $X_1 = \bigcup_{e \in \sce} \varphi^{-1}_N(e)$ 
	and bound the cost of each part individually.
	Clearly we have $w(Z^*) = \OPT$.
	To bound the cost of $X_1$, observe that for some $\varphi_N^{-1}(e) \in X_1$ we have
	$w(\varphi_N^{-1}(e)) \leq \alpha_{j+1} w(e)$, where $\alpha_j$ is the largest $\alpha \in W$ satisfying $\alpha_j \leq \alpha_e$.
	We can bound the cost of $X_1$ by

	\begin{align*}
		w(X_1) & = \sum_{\varphi_N^{-1}(e) \in X_1} w(e) = \sum_{\varphi_N^{-1}(e) \in X_1} \sum_{\text{largest } j: \alpha_e \geq \alpha_j} w(\varphi_N^{-1}(e))\\
		& \leq \sum_{\varphi_N^{-1}(e) \in X_1} \sum_{\text{largest } j: \alpha_e \geq \alpha_j} \alpha_{j+1} w(e) = \sum_{j=0}^N \alpha_{j+1} \beta_j \cdot \OPT(I).
	\end{align*}

	\noindent Combining both bounds we obtain
	\begin{align*}
		w(Y_A) & \leq w(Z^*) + w(X_1) \leq  \left( 1 +  \sum_{j = 0}^{N} 
		\alpha_{j+1} \beta_j \right) \cdot\OPT(I).
	\end{align*}
	\end{proof}
	Since the algorithm computes a $\lambda$-approximation, we have that 
\begin{align*}
	w(Z^A) & \leq \lambda w(Y_A) \leq \left( \lambda + \lambda \cdot \sum_{j = 0}^{N} 
	\alpha_{j+1} \beta_j \right) \cdot\OPT(I).
\end{align*}
\end{proof}

Finally, we bound the cost of the solution output by Algorithm \algoB. 
  \setcounter{myclaim}{0}
\begin{lemma}
  \label{lemma:onerst:refinedbound:B}
  Suppose we run Algorithm~\ref{alg:onerst:approx} with threshold values $\Alphas = \{ \alpha_i \}_{0 \leq i \leq N}$.
  For $1 \leq i \leq N$ let $Z_i^\algoB$ be the solution 
  computed by Algorithm \algoB with scaling factor $\alpha_i$ in Algorithm~\ref{alg:onerst:approx}. Then
  \begin{align*}
    &w(Z_i^\algoB) \leq \left( 1 +  \sum_{j = 0}^{i-1} (\tau -1 + \alpha_{j+1}) \beta_j +  (\tau -1) \cdot \sum_{j=0}^{N} \gamma_j + \sum_{j=i}^{N} \frac{\gamma_j}{\alpha_j} + \sum_{j=0}^{i-1} \gamma_j + (\tau -1) \cdot \xi  \right) \cdot \OPT(I).
  \end{align*}
\end{lemma}
\begin{proof}
  Let $T_i^S$ be the safe edges of the tree $T_i$ and let $\varphi_i: E(T_i)
  \rightarrow E(T)$ be an $\alpha_i$-monotone exchange bijection, where $T$ is
  a spanning tree of the optimal solution $Z^*$ to the instance $(G, w,
  \safe{F})$.
 
  By contracting each edge of $T_i^S$ in $T_i$ and $G$ we obtain the tree $T_i'$ and the 
  graph $G^S_i \coloneqq G / E(T_i^S)$. 
  Algorithm \algoB computes a $\tau$-approximate solution to the instance
  $(G^S_i, T_i')$ of \WTAP.
  \setcounter{myclaim}{0}
  \begin{myclaim}
    \label{claim:H':feasibility3}
    The set
    \[
      Y_i \coloneqq ( \sce \setminus \varphi_i (E(T_i^S))) \cup \bigcup_{1 \leq j \leq r} E(C_j)
    \]
    of edges is a feasible solution to the \WTAP instance $(G^S_i. T_i')$.
  \end{myclaim}
  \begin{proof}
    Clearly $(V, Y_i)$ is a connected graph. 
    It remains to argue that each edge $e \in Y_i$ is contained in some cycle. This is
    certainly true for each edge $e$ of a component $C_j$, $1 \leq j \leq r$.
    It remains to show that the edges in $\sce \setminus
    \varphi_i (E(T_i^S))$ are contained in some cycle
    of $G^S_i$. Since $\varphi_i$ is an exchange bijection, an edge $e \in \sce \setminus \varphi_i(E(T_i^S))$ and its preimage $\varphi^{-1}(e)$ are on a cycle in $E(T) \cup \{e \}$.
    This concludes the proof.
  \end{proof}
  
  We now bound the cost of $Z_i^\algoB$.  By Claim \ref{claim:H':feasibility3} we
  have that $T_i \cup Y_i$ is a feasible solution to the \onerst instance
  $(G, w, \safe{F})$. 
  The algorithm then returns in polynomial time a
  solution $Z_i^\algoB$ of cost at most $w(Z_i^\algoB) \leq w(T_i) + \tau w(Y_i)$.  We
  first bound the cost of each edge of $T_i$ as follows.
  \begin{myclaim}
    \label{claim:T_i:boundcost3}
    Let $e \in E(T_i)$ and let $\alpha_e$ be its threshold value. Then we have
    \[
      w(e) \leq 
      \begin{cases}
	\frac{1}{\alpha_j} \cdot w( \varphi_i (e)) & \text{if $\varphi_i(e) \notin \sce$,}\\
	\alpha_{j+1} \cdot w(\varphi_i(e)) & \text{if } \varphi_i(e)\in \sce \text{ and } e \in F\text{, and }\\
	w( \varphi_i(e))  & \text{otherwise,}\enspace
      \end{cases}
    \]
    where $\alpha_j$ is the largest $\alpha \in \Alphas$ satisfying $\alpha \leq \alpha_e$. 
  \end{myclaim}

    \begin{proof}
It remains to prove $w(e) \leq \alpha_{j+1} w(\varphi_i(e))$ if $\varphi_i(e)\in \sce$ and $e \in F$ since the other part is proven in Claim~\ref{claim:T_i^S:boundcost2} of Lemma \ref{lemma:onerst:refinedbound:C}. Since $e$ is an unsafe edge and $\varphi_i(e) \in E'$, we have that $w(e) \leq \alpha_e w(\varphi_i(e))$. Since $e$ is contained in $T_j$, but not contained in $T_{j+1}$ (by the definition of $j$), we have that $\alpha_e \leq \alpha_{j+1}$.
  \end{proof} 
Recall that by Lemma \ref{lemma:onerst:refinedbound:C} we have 
\[
w(T_i^S) \leq \left(\sum_{j=i}^{N} \beta_{j}+ \sum_{j=i}^{N} \frac{\gamma_j}{\alpha_j}\right) \cdot \OPT(I)
\]
According to Claim~\ref{claim:T_i:boundcost3} we have for $1 \leq i \leq N$ that
    \begin{align*}
      w(T_i \cap F) & = \sum_{e \in T_i \cap F} w(e)\\
      & = \sum_{e \in T_i \cap F} \sum_{\text{largest } j: \alpha_e \geq \alpha_j} w(e)\\
      & = \sum_{e \in T_i \cap F : \varphi_i(e) \in E'} \sum_{\text{largest } j: \alpha_e \geq \alpha_j} w(e) + \sum_{e \in T_i \cap F : \varphi_i(e) \notin E'} \sum_{\text{largest } j: \alpha_e \geq \alpha_j} w(e)\\
      & \overset{\text{Claim \ref{claim:T_i:boundcost3}}}{\leq} \sum_{e \in T_i \cap F : \varphi_i(e) \in E'} \sum_{\text{largest } j: \alpha_e \geq \alpha_j} \alpha_{j+1} w(\varphi_i(e)) + \sum_{e \in T_i \cap F : \varphi_i(e) \notin E'} \sum_{\text{largest } j: \alpha_e \geq \alpha_j} w(\varphi_i(e))\\
      & \leq \left(\sum_{j=0}^{i-1} \alpha_{j+1} \beta_{j} + \sum_{j=0}^{i-1} \gamma_j\right) \cdot \OPT(I)\enspace.
    \end{align*}
Thus we have 
\begin{align*}
    w(T_i) & = w(T_i^S) + w(T_i \cap F) \leq \left(\sum_{j=0}^{i-1} \alpha_{j+1} \beta_{j} + \sum_{j=i}^{N} \beta_{j}  + \sum_{j=0}^{i-1} \gamma_j + \sum_{j=i}^{N} \frac{\gamma_j}{\alpha_j}\right) \cdot \OPT(I) \\
    		& = \left( 1 + \sum_{j=0}^{i-1} (\alpha_{j+1} -1) \beta_{j} - \sum_{j=0}^{N} \gamma_j - \xi  + \sum_{j=0}^{i-1} \gamma_j + \sum_{j=i}^{N} \frac{\gamma_j}{\alpha_j}\right) \cdot \OPT(I)\enspace,
\end{align*}
where the last equality holds due to $\xi + \sum_{j = 0}^N \beta_j + \sum_{j=0}^N \gamma_j = 1$.
We additionally bound the cost of $Y_i$.
  \begin{myclaim}
    \label{claim:Y_i:boundcost3}
  \begin{align*}
    w(Y_{i}) \leq \left( \sum_{j=0}^{i-1} \beta_j + \sum_{j=0}^N \gamma_j + \xi \right) \cdot\OPT(I)\enspace.
  \end{align*}
\end{myclaim}
\begin{proof}
We need to bound the cost of
    \[
      Y_i = \left( \sce \setminus \varphi_i (E(T_i^S)) \right) \cup \bigcup_{1 \leq j \leq r} E(C_j).
    \]
We first bound the cost of $\sce \setminus \varphi_i (E(T_i^S))$.
According to the definition of $\beta_0, \beta_1, \ldots, \beta_N$ we can bound the cost of 
$\bigcup_{e \in \sce \setminus \varphi_i (E(T_i^S))} e$ by $\sum_{j = 1}^{i-1} \beta_j$.
Additionally we bound the cost of $\bigcup_{1 \leq j \leq r} E(C_j)$ 
by $(\xi + \sum_{j=1}^N \gamma_j )\cdot \OPT$.
Putting things together we obtain 
  \begin{align*}
    w(Y_{i}) \leq \left( \sum_{j=0}^{i-1} \beta_j + \sum_{j=0}^N \gamma_j + \xi \right) \cdot\OPT(I)\enspace.
  \end{align*}
\end{proof}
Finally, since the algorithm computes a $\tau$-approximate solution for the \WTAP instance,
we have
\begin{align*}
      w(Z_i^\algoB) & \leq w(T_i) + \tau \cdot w(Y_i)\\
      & \leq \left( 1 + \sum_{j=0}^{i-1} (\alpha_{j+1} -1) \beta_{j} - \sum_{j=0}^{N} \gamma_j - \xi  + \sum_{j=0}^{i-1} \gamma_j + \sum_{j=i}^{N} \frac{\gamma_j}{\alpha_j}\right) + \left( \sum_{j=0}^{i-1} \beta_j + \sum_{j=0}^N \gamma_j + \xi \right) \cdot \OPT(I)\\
      & \leq \left( 1 +  \sum_{j = 0}^{i-1} (\tau -1 + \alpha_{j+1}) \beta_j +  (\tau -1) \cdot \sum_{j=0}^{N} \gamma_j + \sum_{j=i}^{N} \frac{\gamma_j}{\alpha_j} + \sum_{j=0}^{i-1} \gamma_j + (\tau -1) \cdot \xi  \right) \cdot \OPT(I).
\end{align*}
  which concludes the proof.
\end{proof}

With a similar argument as in Section \ref{sec:onerst:5/2} we can argue that
$\xi = 0$.  
For completeness we state it here again.
  Since $\lambda, \tau \geq 1$, each $\beta_j, \gamma_j, j \in [N]$ as well as $\xi$
  has a positive coefficient in the bounds of the lemmas
  \ref{lemma:onerst:refinedbound:C}, \ref{lemma:onerst:refinedbound:A} and \ref{lemma:onerst:refinedbound:B}.
  Moreover, since we maximize over $\beta_1, \beta_2, \ldots, \beta_N$,
  $\gamma_1, \gamma_2, \ldots, \gamma_N$ and $\xi$ we may assume $\xi = 0$.  To
  see this, suppose we have an optimal choice of the variables where $\xi > 0$.
  Then, consider the following new variables.  Let $\beta_j' = \beta_j$ for $j
  \in [N], \gamma_j' = \gamma_j$ for $j \in [N-1]$, $\gamma_N' = \gamma_N + \xi$
  and $\xi' = 0$. Now observe that the value of the minimum over $w(Z^\algoA)$ as well as
  $w(Z^\algoC_i)$ and $w(Z^\algoB_i), i \in [N]$ is at least as large for the new variables as for
  the old ones.  Thus we can assume that $\xi = 0$ and have 
  \begin{equation*}
    \sum_{j=1}^{N} \beta_j + \sum_{j=1}^{N} \gamma_j = 1.
  \end{equation*}
Hence we found suitable functions for $f^\algoA(\cdot),\,
f_i^\algoB(\cdot),\, f_i^\algoC(\cdot)$.  These are given in lemmas
\ref{lemma:onerst:refinedbound:A}, \ref{lemma:onerst:refinedbound:B} and
\ref{lemma:onerst:refinedbound:C}, respectively (with $\xi$ set to 0). 

\begin{proof}[Computational proof of Theorem \ref{thm:flexbounded}]
We choose $N = 60$ and let $\tau = 1.5$ and $\lambda = 2$ and 
\begin{align*}
    (\alpha_i)_{0 \leq i \leq N+1} &= 0, 0.05, 0.1, 0.15, 0.2, 0.25, 0.3, 0.35, 0.4, 0.45, 0.5, 0.51, 0.52, 0.53, 0.54, 0.55, 0.56, 0.57, \\ 
   & 0.58, 0.59, 0.60, 0.61, 0.62, 0.63, 0.64, 0.65, 0.67, 0.68, 0.69, 0.7, 0.71, 0.72, 0.73, 0.74, 0.75, 0.76, \\
   & 0.77, 0.78, 0.79, 0.8, 0.81, 0.82, 0.83, 0.84, 0.85, 0.86, 0.87, 0.89, 0.9, 0.91, 0.92, 0.93, 0.94, 0.95,\\
      & 0.96, 0.97, 0.98, 0.99, 1 \enspace.
\end{align*} 
With these choices we obtain from~\eqref{eq:minmaxmin} the following optimization problem.
  \begin{equation}
  \label{optforbound}
  \begin{aligned}
   \text{maximize } & z \\
   \text{subject to } & z \leq  \left( \lambda + \lambda \cdot \sum_{j = 0}^{N} 
     \alpha_{j+1} \beta_j \right) \\
   & z \leq \left( 1 +  \sum_{j = 0}^{i-1} (\tau -1 + \alpha_{j+1}) \beta_j +  (\tau -1) \cdot \sum_{j=0}^{N} \gamma_j + \sum_{j=i}^{N} \frac{\gamma_j}{\alpha_j} + \sum_{j=0}^{i-1} \gamma_j \right) \text{ for } 1 \leq i \leq N\\
   & z \leq \left( 1 +  \sum_{j = 0}^{i-1} (\lambda -1 + \lambda \alpha_{j+1}) \beta_j +  (\lambda -1) \cdot \sum_{j=0}^{N} \gamma_j + \sum_{j=i}^{N} \frac{\gamma_j}{\alpha_j} \right) \text{ for } 1 \leq i \leq N\\
   & \sum_{j = 1}^{N} \beta_j + \gamma_j = 1\\
   & z \geq 1 \\
   &\beta_j, \gamma_j \in [0, 1] \text{ for } 1 \leq j \leq N.
   \end{aligned}
  \end{equation}

  The constraints are given by the bounds from
  lemmas~\ref{lemma:onerst:refinedbound:C} --
  \ref{lemma:onerst:refinedbound:B}. The solution of the problem is an upper
  bound on the approximation ratio of Algorithm \ref{alg:onerst:approx} due to
  Lemma \ref{lemma:onerst:bestalpha}.
  We obtain a computational upper bound of \boundthreehalf and lower bound of
  \boundthreehalflow on the optimal value of the non-linear program above using
  the NLP solver baron~\cite{baron}. 
\end{proof}

\section{Algorithms~\algoA and~\algoB give a $14/5$-approximation}
\label{sec:onerst:2.8}

In this section we give an analytic upper bound of $14/5$ on the approximation
ratio of Algorithm~\ref{alg:onerst:approx}. In our analysis we consider just
the contribution of algorighms \algoA and \algoB and ignore the contribution of Algorithm \algoC.
We obtain an analytic upper bound on the value of the min-max-min optimization problem
\eqref{eq:minmaxmin} in terms of $\tau$ and $\lambda$, which gives in turn an upper bound on the
approximation ratio of Algorithm~\ref{alg:onerst:approx}. Note that the
combined worst-case of algorithma \algoA and \algoC is much better than their
individual ratios.
One of the key insights is that the two algorithms exhibit their worst-case performance if
Algorithm \algoB runs with a scaling factor of $1/2$.

Let us fix some $\alpha \in [0, 1]$ and let $T_\alpha$ be an $\alpha$-MST of $G$.
We consider the graph $G_\alpha = (V', E_\alpha)$ obtained from $T_\alpha$ by
identifying for $1 \leq j \leq r$ the vertex set of the 2-edge-connected
component $C_j$ of the optimal solution $Z^*$ with a single vertex, discarding
loops but not parallel edges. 
Since $T_\alpha$ is a tree, the graph $G_\alpha$ is connected. Let $T$ be a
spanning tree of $Z^*$ and let $\varphi: T_\alpha \rightarrow T$ be an
$\alpha$-monotone exchange bijection, which exists according to
Lemma~\ref{lem:onerst:existence:alpha:mon:exchange:bijection}.  Since every
spanning tree of $Z^*$ contains $\sce$ we have that  $\sce \subseteq \varphi
(E_\alpha)$.

We partition the edge set of $T_\alpha$ into four parts $D_\alpha$,
$O_\alpha$, $F_\alpha$, and $S_\alpha$ as follows.
\begin{itemize}
    \item $D_\alpha \coloneqq \{ e \in E(T_\alpha) \cap F \mid \varphi(e) \in E' \}$
    \item $O_\alpha \coloneqq \{ e \in E(T_\alpha) \cap \safe{F} \mid \varphi(e) \in E' \}$
    \item $F_\alpha \coloneqq \{ e \in E(T_\alpha) \cap F \mid \varphi(e) \in E(T) \setminus E' \}$
    \item $S_\alpha \coloneqq \{ e \in E(T_\alpha) \cap \safe{F} \mid \varphi(e) \in E(T) \setminus E' \}$
\end{itemize}

Note that the partition of $E(T_\alpha)$
depends on the optimal solution $Z^*$ and the spanning
tree $T$ of $Z^*$, so we cannot expect to compute them efficiently. 
We now define the following variables. 
\begin{align*}
& b_\alpha \coloneqq \frac{w(\varphi(O_\alpha))}{\OPT(I)} \text{ and } c_\alpha \coloneqq \frac{w(\varphi(S_\alpha))}{\OPT(I)}\\
& b_0 \coloneqq \frac{w(\sce) - w(\varphi(O_\alpha))}{\OPT(I)} \text{ and } c_0 \coloneqq \frac{w(T - \sce) - w(\varphi(S_\alpha))}{\OPT(I)}\\
& \xi = \frac{\OPT(I) - w(T)}{w(\OPT(I)}
\end{align*}
So $b_\alpha$ (resp., $c_\alpha$) is the fraction of the weight of $\OPT(I)$
of the safe edges in $T_\alpha$ that are mapped to safe cut edges in $T$ 
(repsp., edges of $E(T) - \sce$).
The value $b_0$ (resp., $c_0$) represent the fraction of the weight of $\OPT(I)$
of \unsafe edges in $T_\alpha$ that are charged to safe cut edges in $T$ 
(resp., edges of $E(T) - \sce$).
Finally, $\xi$ is the fraction of the weight of $\OPT(I)$ that is not contributed by $E(T)$. 

Observe that the following holds. 
\begin{enumerate}
	\item $b_\alpha, c_\alpha, b_0, c_0, \xi \in [0, 1]$,
	\item $b_\alpha + b_0 = w(\sce)/\OPT(I)$,
	\item $c_\alpha + c_0 = w(T - \sce)/{\OPT(I)}$, and
	\item $b_0 + b_\alpha + c_0 + c_\alpha + \xi = 1$.
\end{enumerate}

We use the properties of $\alpha$-monotone exchange bijections to show that the
minimum weight of the solution computed by Algorithm \algoB run with scaling
factor $\alpha$ and the solution computed by Algorithm \algoA can be bounded in
terms of $b_\alpha, b_0$, $c_\alpha, c_0, \xi$, and $\OPT(I)$ as follows.

\begin{lemma}
  \label{lemma:approx:cost}
  Suppose we run Algorithm~\ref{alg:onerst:approx} with a single scaling factor
  $\alpha \in [0,1]$ (i.e., $W = \{ \alpha \}$) and let $Z_\alpha^B$ be the solution computed by
  Algorithm $\algoB$ with scaling factor $\alpha$ in Algorithm~\ref{alg:onerst:approx}. Then 
  \[
        w(Z_\alpha^B) \leq \left( (\tau + \alpha) b_0 + b_\alpha + (\tau + 1) c_0 + (\tau + \frac{1}{\alpha}) c_\alpha + \tau \cdot \xi \right) \cdot\OPT(I)\enspace.
  \]
\end{lemma}
\begin{proof}
  Let $T_\alpha$ be the $\alpha$-MST computed by Algorithm \algoB and let
  $H_\alpha$ be the solution to the \WTAP instance $(G, T_\alpha', w)$
  computed by Algorithm \algoB. The following bound on $w(T_\alpha)$
  follows in a straight-forward manner from Definition~\ref{def:alphaexchange}.
  \setcounter{myclaim}{0}
  \begin{myclaim}
    We have $w(T_\alpha) \leq \left( \alpha b_0 + b_\alpha + c_0 + c_\alpha / \alpha \right) \cdot \OPT(I)$.
    \label{claim:treecost}
  \end{myclaim}
  \begin{proof}
    Since $T_\alpha$ is an $\alpha$-MST and $\varphi : T_\alpha \rightarrow T$
    is an $\alpha$-monotone exchange bijection, we have that for each edge $e$ of
    $T_\alpha$ exactly one of the following cases applies.
    \begin{enumerate}
      \item If $e \in D_\alpha$ then we have $w(e) \leq \alpha \cdot\varphi(e)$, or
      \item if $e \in O_\alpha \cup F_\alpha$ then we have $w(e) \leq w(\varphi(e))$, or
      \item if $e \in S_\alpha$ then $w(e) \leq \frac{1}{\alpha} w(\varphi(e))$.
    \end{enumerate}
    
    By summing over the above inequalities and applying the definition of $b_0, b_\alpha, c_0, c_\alpha$ we have
\begin{align*}
      w(T_\alpha) & \leq \left( \alpha \cdot w(\varphi (D_\alpha)) + w(\sce \setminus \varphi (D_\alpha)) + \varphi(S_\alpha) / \alpha + \varphi(F_\alpha) \right) \cdot \OPT(I)\\
      & \leq \left( \alpha \cdot b_0 + b_\alpha + c_0 + c_\alpha / \alpha \right) \cdot \OPT(I)
\end{align*}
    as claimed.
  \end{proof}

  It remains to bound the cost of $H_\alpha$. 
  Note that $H_\alpha$ is a $\tau$-approximate solution of an optimal solution for the
  augmentation problem.
  To bound the cost of an optimal augmentation we demonstrate the existence of a feasible augmentation.
  We show that $Y_\alpha \coloneqq \varphi(D_\alpha) \cup \varphi^{-1}(Z^* - \sce)$ is such a feasible augmentation. 
  \begin{myclaim}
    \label{claim:augmentationcost}
    $T_\alpha \cup Y_\alpha$ is a feasible solution to $I$.
  \end{myclaim}
  \begin{proof}
  Clearly $E(T_\alpha) \cup Y_\alpha$ is a connected spanning subgraph since it contains the tree $T_\alpha$.
  	It remains to show that $E(T_\alpha) \cup Y_\alpha -f$ is connected for each $f \in F \cap (E(T_\alpha) \cup Y_\alpha)$.
  	This is clearly true for the edges in $F_\alpha$ since we added $\varphi^{-1}(Z^* - \sce)$.
  	Thus let $f \in D_\alpha$. Since $f \in F$ and $\varphi(f) \in \safe{F}$, we have that $f \neq \varphi(f)$.
  	By the definition of exchange bijections we then have that there is a cycle in $T_\alpha \cup \varphi(f)$ going through $f$. This concludes the proof.
  \end{proof}
By the definition of $b_0, b_\alpha, c_0, c_\alpha$ and $\xi$ we have that
\begin{align*}
w(Y_\alpha) & \leq w(\varphi(D_\alpha)) + w(Z^* - \sce)\\
& \leq w(\varphi(D_\alpha)) + w(\varphi(F_\alpha)) + w(\varphi(S_\alpha)) + \xi \\
& \leq \left( b_0 +  c_0 + c_\alpha + \xi \right) \cdot \OPT(I).
\end{align*}

  Since Algorithm~\ref{alg:onerst:approx} uses a $\tau$-approximation for \tap,
  we have that $ w(H_\alpha) \leq \tau \cdot w(Y_\alpha)$.

  Using these bounds and claims~\ref{claim:treecost} and~\ref{claim:augmentationcost}, we obtain
  \[
        w(Z_\alpha^B) \leq \left( (\tau + \alpha) b_0 + b_\alpha + (\tau + 1) c_0 + (\tau + \frac{1}{\alpha}) c_\alpha + \tau \cdot \xi \right) \cdot\OPT(I)\enspace.
  \]
  as claimed.
\end{proof}

Next, we provide a bound on the cost of the solution returned by Algorithm \algoA.
\begin{lemma}
  \label{lemma:approx:cost:algoA}
  Let $\alpha \in [0, 1]$ and let $Z^A$ be a solution computed by Algorithm $\algoA$ in Algorithm~\ref{alg:onerst:approx}. Then 
  \[
        w(Z^A) \leq \left( \lambda + \lambda \cdot (\alpha \cdot b_0 + b_\alpha) \right) \cdot\OPT(I)\enspace.
  \]
\end{lemma}

\begin{proof}
We demonstrate the existence of a feasible solution $Y \subseteq E$ to $I$.
Recap that Algorithm \algoA computes a 2-ECSS solution to $I$ where each safe edge $e \in \safe{F}$
has a parallel edge $e'$.
Let $Y\coloneqq Z^* \cup \varphi(D_\alpha) \cup \{e' : e \in \sce \setminus \varphi(D_\alpha) \}$.
Since $Z^* \subseteq Y$, each edge in $Z^* \setminus \sce$ is contained in a 2-edge connected component.
Furthermore, since $Y$ additionally contains $\varphi(D_\alpha) \cup \{e' : e \in \sce \setminus \varphi(D_\alpha) \}$, each edge in $\sce$ is also contained in a 2-edge connected component.
Thus $Y$ is a 2-ECSS. 

It remains to bound the cost of $Y$.
We bound the cost of each part individually.
Clearly we have $w(Z^*) = \OPT(I)$. 
By the definition of $b_0$ and $b_\alpha$, we have $w(\varphi(D_\alpha)) \leq \alpha \cdot b_0 \cdot \OPT(I)$ and
$w(\{e' : e \in \sce \setminus \varphi(D_\alpha) \}) \leq b_\alpha \cdot \OPT(I)$.

Since Algorithm \algoA uses a $\lambda$-approximation for the 2-ECSS problem, we obtain
  \[
        w(Z^A) \leq \left( \lambda + \lambda \cdot (\alpha \cdot b_0 + b_\alpha) \right) \cdot\OPT(I)\enspace.
  \]
\end{proof}

\begin{proposition}
    Algorithm~\ref{alg:onerst:approx} is an approximation algorithm for \onerst with
    ratio 
     \[ 
     \min \{ 1 + \tau, \frac{\lambda \cdot (4 \tau^2 + \sqrt{1+4 \tau} -2 \tau-1)}{(1- \lambda) \cdot \sqrt{1+4 \tau}+2 \tau^2+(2 \lambda -2) \cdot \tau -1 + \lambda} \}\enspace.
     \]
  \label{prop:28general}
\end{proposition}
\begin{proof}
The guarantee of Algorithm \ref{alg:onerst:approx} is clearly bounded by the weight of the solution 
computed by Algorithm \algoB. This weight is at most $1 + \tau$, which is the first part of the minimum.
For the second part we know that $w(Z^\algoA) = w(Z^\algoB_\alpha)$.
By lemmas~\ref{lemma:approx:cost} and \ref{lemma:approx:cost:algoA} we have
  \[
        w(Z_\alpha^B) \leq \left( (\tau + \alpha) b_0 + b_\alpha + (\tau + 1) c_0 + (\tau + \frac{1}{\alpha}) c_\alpha + \tau \cdot \xi \right) \cdot\OPT(I)\enspace \text{ and}
  \]

  \[
        w(Z^A) \leq \left( \lambda + \lambda \cdot (\alpha \cdot b_0 + b_\alpha) \right) \cdot\OPT(I)\enspace.
  \]
  Let $\algo(I)$ be the cost of a solution computed by Algorithm~\ref{alg:onerst:approx}.
  An upper bound on the approximation ratio of Algorithm \ref{alg:onerst:approx} is
\begin{equation}
    \begin{aligned}
          \algo(I) \leq \max_{\substack{b_0, b_\alpha, c_0, c_\alpha, \xi \in [0, 1] }} \;\min_{\alpha \in [0, 1]} \qquad&  \{w(Z^\algoA), w(Z^\algoB_\alpha) \}\\
          \text{subject to}\qquad& b_0 + b_\alpha + c_0 + c_\alpha + \xi = 1
    \end{aligned}
    \label{eq:maxmin:proof}
\end{equation}
  It is easy to see that the maximum is obtained if $c_0 = \xi = 0$. 
  By substituting $c_\alpha = 1 - b_0 - b_\alpha$ we obtain

  \[
	  w(Z_\alpha^B) \leq \left( (\tau + \alpha) b_0 + b_\alpha + (\tau + \frac{1}{\alpha}) (1 - b_0 - b_\alpha) \right) \cdot\OPT(I)\enspace.
  \]
  By putting 
  $w(Z^\algoA) = w(Z^\algoB_\alpha)$ we obtain
  $b_0 = \frac{(-b_\alpha \cdot \lambda - \lambda) \cdot \alpha - \alpha \cdot b_\alpha \cdot \tau+b_\alpha \cdot \alpha + \tau \alpha - b_\alpha + 1}{\alpha^2 \cdot \lambda - \alpha^2 + 1}$.
  Plugging this into $w(Z^\algoA)$ and setting $\alpha = \frac{-1+\sqrt{1+4 \tau}}{2 \tau}$ (recall that we are free to choose $\alpha$) we obtain
  \[ 
	  \algo(I) \leq \frac{\lambda \cdot (4 \tau^2 + \sqrt{1+4 \tau} -2 \tau-1)}{(1- \lambda) \cdot \sqrt{1+4 \tau}+2 \tau^2+(2 \lambda -2) \cdot \tau -1 + \lambda}\enspace.
  \]

\end{proof}

By setting $\lambda = \tau = 2$ we directly obtain the following result.

\begin{corollary}
\label{cor:alg:2.8}
Algorithm~\ref{alg:onerst:approx} is a polynomial-time $14/5$-approximation algorithm for \onerst
for $\tau = \lambda = 2$.
\end{corollary}

\section{Further Results}
\label{sec:further}

In this section we give an approximation algorithm for the unweighted
$k$-\flex, which asks for a minimum cardinality subgraph of a given
graph, such that the removal of any $k$ \emph{unsafe} edges results in a
connected graph.  Note that $1$-\flex is simply \flex. Furthermore, we show
that a generalization of \flex to matroids is \textsc{Set Cover}-hard to
approximate.

\subsection{Unweighted \krst}
\label{sec:unweighted}

Recall that we denote by $\vartheta_k$ the approximation ratio of a
polynomial-time algorithm for the problem unweighted $k$-ECSS.  Currently, the
best values known for $\vartheta_k$ are $5/4$ for $k=2$ due to a result of \c{C}ivril
~\cite{ccivril20195}, $1 + 2/(k+1)$ for $3 \leq k \leq 6$ due to
Cheriyan and Thurimella~\cite{CT:00}, and $1 + 1/2k + O(1/k^2)$ for $k \geq 7$
due to Gabow and Gallagher~\cite{Gabow:12}. In the remainder of this section we
prove the following theorem.

\unweightedkflex*

Note that with the current best value for $\vartheta_2$,
Theorem~\ref{thm:unweightedkflex} gives a $23/16$-approximation guarantee for
unweighed \onerst.  In the following, let $I = (G, \safe{F}, k)$ be an instance
of unweighted \krst. The approximation algorithm proceeds as follows. First, we
compute a maximum forest $X$ on the safe edges $\safe{F}$. We then compute an
$\vartheta_{k+1}$-approximate $(k+1)$-edge connected spanning subgraph $Y$ of
$G/X$ and output the graph $X \cup Y$. 

\begin{lemma}
  \label{lemma:krst:structure}
  Let $H \subseteq G$ be a feasible solution to $I$. Then $H/({\safe{F}\cap
  E(H)})$ is $(k+1)$-edge connected. 
\end{lemma}
\begin{proof}
  Suppose for a contradiction that $G' \coloneqq H / (\safe{F} \cap E(H)$ is at most
  $k$-edge connected. That is, there are two vertices $v, w \in V(G')$ that
  are connected by at most $k$ edge-disjoint paths. Therefore, there is a cut
  $F' \subseteq E(G')$ of size at most $k$ separating $u$ and $v$. But then
  $F'$ is also a cut of size at most $k$ in $H$ and $F'$ consists only of
  \unsafe edges.  Therefore, $H$ is not feasible, a contradiction.
\end{proof}

Let $Z^* \subseteq G$ be an optimal solution to $I$ and let $H = X + Y$ be a solution
computed by the algorithm described above. Suppose that a maximum forest in $G -
F$ has size $\ell = |E(X)|$.
Let $Y^*$ be a minimum $(k+1)$-edge connected
spanning subgraph of $G / E(X)$. We may compute in polynomial-time a $(k+1)$-ECSS of
$G / E(X)$ of size $\vartheta_{k+1} \cdot |E(Y^*)|$.
Therefore, the solution $X + Y$ output by the algorithm has size 
\[
  |E(X)| + |E(Y)| \leq \ell + \vartheta_{k+1} \cdot |E(Y^*)| \leq  \ell + \vartheta_{k+1} \cdot \OPT(I)
\]
On the other hand, we have that
\[
  |E(X)| + |E(Y)| \leq \ell + 2 (k+1) (n- \ell) \leq 2 \OPT(I) - (2k+1) \ell
\]

Hence $|X| + |Y| \leq \min \{\ell + \vartheta_{k+1} \cdot \OPT(I),\,  2 \OPT(I) -
 (2k+1) \ell\} \leq \OPT(I)(2 + \vartheta_{k+1} (2k+1))/(2k+2)$ as claimed.

\subsection{Approximation Hardness on Transversal Matroids}
\label{sec:onerst:hardness}

Our main technical tool in the analysis of our approximation algorithm for
\onerst are exchange bijections, which are based on a matroid basis exchange
argument. So it is natural to ask whether our results can be transferred to a
matroid setting entirely. Let us consider a generalization of \onerst as follows.

\begin{quote}
    \textsc{Flexible Matroid Basis}\\
    \textbf{instance:} matroid $M$ on a ground set $\groundset$, weights $w \in \Z_{\geq 0}^X$, \unsafe items $F \subseteq \groundset$\\
    \textbf{task:} Find minimum-weight $E \subseteq \groundset$ such that for each $f \in F$, the set $E-f$ contains a basis of $M$.
\end{quote}

Observe that for graphic matroids this problem corresponds to \onerst.
Please note also that $\alpha$-MSTs and their threshold properties
 as well as $\alpha$-monotone
exchange bijections (Section~\ref{sec:onerst:technical}) generalize in
a natural way to matroids by replacing ``spanning tree'' by ``matroid basis''
in the respective lemmas. We show that despite these promising observations,
under standard complexity assumption, we cannot hope for a polynomial-time
constant-factor approximation algorithm for the problem \textsc{Flexible
Matroid Basis}. 
Let $G = (U,
V, E)$ a bipartite graph and let $\mathcal{I} \coloneqq \{ F \subseteq U \mid
\text{there is a matching $M$ of $G$ that covers $F$}\}$. Then $\mathcal{I}$ is
the set of independent sets of a matroid. Matroids that can be obtained in this
manner are called \emph{transversal matroids}, see for instance~\cite{Oxley}
for an introduction to the theory of transversal matroids. The next theorem shows that
\textsc{Flexible Matroid Basis} on transversal matroids is as hard to
approximate as \textsc{Set Cover}.

\begin{theorem}
  \label{thm:onerst:hardness}
  \textsc{Flexible Matroid Basis} on transversal matroids admits no
  polynomial-time $(1-\varepsilon)\log |X|$-approximation algorithm for
  any $\varepsilon > 0$ unless $\P = \NP$. 
\end{theorem}
\begin{proof}
    Let $I = (U, \mathcal{S})$ be an instance of \textsc{Set Cover}, where
    $U = \{ u_1, u_2, \ldots, u_n \}$ is the set of items to be covered and
    $\mathcal{S} = \{ S_1, S_2, \ldots, S_m \}$ is a familiy of subsets of $U$.
    We construct an instance $I \coloneqq (M(G), w, F)$ of \textsc{Flexible Matroid Basis} as follows.
    Let  $G = (A, B, E)$ be a bipartite graph given by
    \begin{align*}
        A \coloneqq & \{u_1, u_2, \ldots, u_n \} \cup \{ S_1, S_2, \ldots, S_m \} \\
        B \coloneqq & \{v_1, v_2, \ldots, v_n \} \\
        E \coloneqq & \{ S_iv_j \mid u_j \in S_i,\,1 \leq i \leq m,\,1 \leq j \leq n\}\enspace.
    \end{align*}
    Let  $M(G)$ be the transversal matroid that arises from the subsets of $A$
    that are covered by a matchings of $G$.
    Clearly, the set $U$ is a basis of $M(G)$. Let $F \coloneqq A \setminus U$ be the
    set of \unsafe items and let the weights $w \in \Z_{\geq 0}^A$ be given by
    $w(e) = 0$ if $e \in U$ and $w(e) = 1$ otherwise. This concludes the
    construction of the instance $I' = (M(G), w, F)$ of \textsc{Flexible Matroid Basis}.
    
    Let $Z^* \subseteq A$ be a
    optimal solution to the instance $I'$. We may assume that $Z^*$ contains $U$. We claim that $Z^* \setminus U$ is an optimal solution
    to the instance $I$. First, suppose that $Z^* \setminus U$ is not a cover. Then there
    is an item $u \in U$ that is not covered by $Z^* \setminus U$. But then $Z^* -
    u$ does not contain a basis of $M(G)$. Now suppose that $Z^* \setminus U$ is
    not optimal for $I$. Then there is a familiy $C \subseteq \mathcal{S}$ of
    sets that cover $U$, which is cheaper than $Z^* \setminus U$. But then $U
    \cup C$ is cheaper solution to $I'$ than $Z^*$, which contradicts the optimality of
    $Z^*$. We conclude that the values of optimal solutions to $I$ and $I'$ are the same.
    Therefore, any
    polynomial-time $\rho$-approximation algorithm for \textsc{Robust
    Matroid Basis} gives a polynomial-time $\rho$-approximation
    algorithm for \textsc{Set Cover}. Hence, the approximation hardness result
    for \textsc{Set Cover} by Dinur and Steurer~\cite{DS:14} implies that there is no
    polynomial-time $(1-\varepsilon)\log |X|$-approximation algorithm
    for \textsc{Robust Spanning Tree} unless $\P = \NP$.
\end{proof}

By a theorem of Piff and Welsh~\cite{PW:70}, any transversal matroid is
representable over any sufficiently large field. Therefore, \textsc{Flexible
Matroid Basis} is \textsc{Set Cover}-hard to approximate on vector matroids.
Note that in the proof of Theorem~\ref{thm:onerst:hardness}, we essentially
show approximation hardness of a generalization of \WTAP to matroids. It is an
interesting open question, whether this generalization admits constant-factor
approximation algorithms for subclasses of transversal or vector matroids, for instance,
regular, binary, and bicircular  matroids.

\bibliography{references}
\appendix

\end{document}